\documentclass[11pt,fleqn]{article}

\usepackage{subcaption}

\usepackage{amsfonts}
\usepackage{amsmath,amssymb}
\usepackage{amsthm}
\usepackage{algorithmic}
\usepackage{array}
\usepackage{epsf}
\usepackage{graphicx,psfrag,color,pst-grad,pstcol}
\usepackage{latexsym}
\usepackage{calc}
\usepackage{multicol}

\def\qed{\hfill \mbox{\rule[0pt]{1.5ex}{1.5ex}} \vskip.1cm}
\def\R{\mathbb{R}}

\def\cL{{\cal L}}

\def\cT{{\cal T}}

\def\cB{{\cal B}}

\def\cP{{\cal P}}

\def\cM{{\cal M}}

\def\cQ{{\cal Q}}

\def\qed{\hfill \mbox{\rule[0pt]{1.5ex}{1.5ex}}\vskip.3cm}

\newtheorem{remark}{Remark}
\newtheorem{definition}{Definition}
\newtheorem{theorem}{Theorem}

\newtheorem{prop}{Proposition}
\newtheorem{property}{Property}

\newtheorem{assumption}{Assumption}

\begin{document}

\def\BibTeX{{\rm B\kern-.05em{\sc i\kern-.025em b}\kern-.08em
    T\kern-.1667em\lower.7ex\hbox{E}\kern-.125emX}}

\title{\LARGE \bf Mixed potential for nonlinear RLC circuits with memristors}

\author{Mauro Di Marco, Mauro Forti, Luca Pancioni,\\ Giacomo Innocenti, Alberto Tesi
\thanks{M. Di Marco, M. Forti and L. Pancioni are with the Department of Information Engineering and Mathematics, University of Siena, v. Roma 56 - 53100 Siena, Italy, e--mail: {\tt\small mauro.dimarco@unisi.it, mauro.forti@unisi.it, luca.pancioni@unisi.it}.
G. Innocenti and A. Tesi are with the Department of Information Engineering,
University of Florence, via S. Marta 3 - 50139 Firenze, Italy, e--mail: {\tt\small giacomo.innocenti@unifi.it,alberto.tesi@unifi.it}.
}
}

\maketitle

\begin{abstract}
In two seminal articles published in 1964, Brayton and Moser introduced the
concept of a \emph{mixed potential} as a fundamental
theoretic tool to describe and analyze a class RLC of nonlinear
circuits containing resistors, capacitors and inductors. In this paper, it is shown
for the first time that a mixed potential can be introduced for a class RLCM of RLC circuits
containing also memristors. This is possible provided a memristor circuit is analyzed
not in the traditional voltage-current domain but rather
in the flux-charge domain. The flux-charge analysis method (FCAM) plays a crucial
role in the extension, in particular, a key step is an equivalence principle established
via FCAM between an RLCM circuit in the flux-charge domain and a nonlinear RLC
circuit in the voltage-current domain. Several examples are discussed where
the mixed potential is explicitly found. These include basic circuits with memristors,
such as Chua's circuit with a memristor and also large-scale memristor arrays with
a neural architecture. This paper is mainly devoted to the introduction of a mixed
potential for memristor circuits and the study of its main theoretic properties,
as the possibility to write the circuit state equations in the flux-charge domain
in an effective and compact form via the mixed potential. In
a companion paper \cite{DiMarco2026BraytonAppl}, the mixed potential is used to
obtain in a systematic way Lyapunov-like results on convergence of RLCM circuits. Those results
will extend existing results on convergence that do not cover the important
case where there is the simultaneous
presence of capacitors and inductors in a memristor circuit.
\end{abstract}

{\bf keywords:} Convergence, flux-charge analysis method,
invariant of motion, memristor, mixed potential, nonlinear circuit, nonlinear dynamics.

\section{Introduction}

Memristor has been devised by L. Chua as the fourth basic passive circuit element
in addition to the resistor, capacitor and inductor in a seminal paper published
in 1971 \cite{Chua1971}. However, about 40 years passed to arrive at the discovery of memristive
behavior at nanoscale by a team led by S. Williams \cite{Williams2008}.
Nowadays, emerging electronic devices
such as memristors are expected to play an ever increasing role in modern electronics
and neuromorphic circuits \cite{boybat2018neuromorphic,huang2021editorial,Sirakoulis2022717,yang2022ResProgMem,Shao2025BurstingJosephson,Sun2025ReviewNeuromorphicC}. The possibility to program the memristor
conductance makes them tailor made to implement neuron interconnections.
Furthermore, since memristors are able to keep in memory their final state
also when power is disconnected, they are suitable to implement effective
\emph{in-memory} computing schemes where the processing and memorization are at
the same physical location \cite{ielmini2020device,ascoli2020theoretical,10518003,SpecIssue-InMemory2023,11053213,Xiao20242294TCASI,Yang20254127}. Such a property appears to be crucial to
circumvent the so-called von Neumann bottleneck of traditional
computer architectures originating from data that
are continuously exchanged between the central processing unit (CPU)
and the memory, e.g., a random access memory (RAM). From a broader
perspective, memristor circuits are believed to excel at mimicking the
parallel and efficient processing capabilities of brain-like nervous
systems \cite{Xiaoqi2024JMSE,Shan20241815EDL,Zou2025AEM,Sun2025LightSensing,Shi2025NeuronalIIordMem}.

A systematic technique has been developed for the study of nonlinear circuits
containing memristors in \cite{Corinto-Forti-I,cfc2020}. This is based on the analysis in
the integral domain given by the flux and charge rather than the
traditional domain given by voltage and current. Thus, the name Flux-Charge
analysis method (or, FCAM, for short). The method proved effective to
highlight relevant peculiar dynamical features of memristor circuits, such as the
existence for structural reasons of invariants of motion for the dynamic
equations in the voltage-current domain (VCD). As a consequence, the
state space in the VCD can be foliated in invariant manifolds where a
memristor circuit satisfies a manifold-dependent reduced-order dynamics
that can be completely characterized in the flux-charge domain (FCD).
In turn, this implies that a memristor circuits can
display infinitely many different coexisting dynamics and attractors for
a fixed set of circuit parameters and nonlinearities, a
property termed in the literature extreme multistability \cite{Corinto-Forti-II}.
Peculiar to this type of dynamics is also the presence of bifurcations due
to changing the initial conditions for fixed circuit parameters and
nonlinearities (bifurcations without parameters).

FCAM has been applied
for a detailed analysis of the bifurcations and oscillatory phenomena of several basic circuits
as Murali-Lakshmanan-Chua circuit and Chua's circuit with a memristor
\cite{cfc2020,Corinto-Forti-II}. Other contributions
concern the analysis of convergence towards equilibrium points of
a number of neural network
architectures using memristors in the cells or the interconnections \cite{di2017memristor,7516733,DiMarco2020LQprog,di2022convergence,10144925,DiMarco2025convArxiv}.
In particular, a recent paper \cite{di2025robust} used FCAM
to establish systematic and robust results on convergence towards equilibrium points for a broad
class of nonlinear circuits, termed RCM, containing \emph{three basic
circuit elements}, i.e., resistors, capacitors and memristors. It is worth
to note that, to our knowledge, no contributions in the literature address
in a systematic way convergence when capacitors and inductors are simultaneously present.
We refer the reader to \cite{chen2019flux,AlChawa2021compact,Soundararajan2025output}
and references cited therein for other applications and extensions of FCAM.

The dynamic analysis of nonlinear RLC circuits (without memristors) has been one
of the main topics in circuit theory since the sixties
\cite{Desoer1972SpecialIssue,Willson1975nonlinear,Liu1980special,chua1980,chua-book}.
Here, we are mainly interested in the theory developed by Brayton and Moser
in the two seminal papers \cite{brayton1964theoryI,brayton1964theoryII}
to describe and analyze a large class of
nonlinear RLC circuits containing resistors, inductors and capacitors. One main
result is that, when the capacitor voltages and inductor currents are
 a complete set of variables, then the state equations can be described via a suitable
function termed \emph{mixed potential}. This result, in addition
to being of great
theoretic significance, paves the way to obtain nice Lyapunov-like results on
convergence towards equilibrium points for nonlinear
RLC circuits \cite{brayton1964theoryII}. The mixed potential theory has been later generalized and applied in several fields such as stability analysis of power networks and DC microgrids \cite{LiuTSC2023-BMLarge,ChenEle2022-LyapMixed}, passivity-based control of power electronic systems \cite{KosarayuTAC2021}, distributed and large-scale system modeling \cite{Macchelli2016brayton}, stability and control in modern power grids \cite{ChangTSG2021Region} and optimization of non-electrical systems \cite{deRinaldis2005electrmech,kosarajuCSL2017stability}. However, to our knowledge, no attempt has been done so far to extend the mixed potential to classes of nonlinear circuits containing
memristors.

Here, we consider a class of nonlinear circuits, termed RLCM, containing
\emph{all four basic circuit elements}, i.e., resistors, inductors, capacitors
and memristors. On one hand, this class generalizes the class RLC considered by Brayton
and Moser \cite{brayton1964theoryI} by allowing for the presence of memristors. On the other hand, it also
extends the class RCM studied in \cite{di2025robust}, where only one type of reactive elements is taken into account.
Considering inductors in a circuit is of course of great theoretic interest, since,
as it is well-known, the simultaneous presence of both electric and magnetic field
greatly enriches the
dynamical behavior. Furthermore, from a more practical viewpoint,
small parasitic inductors need to be
necessarily accounted for in any electronic
implementation, especially when the circuits operate at
high frequencies.

The main contributions in the paper can be summarized as follows:

(a) we extend the
fundamental notion of mixed potential, originally introduced by Brayton and
Moser for RLC circuits without memristors, to the considered class RLCM
of memristor circuits.
This extension is possible provided we represent the memristor circuits in the FCD
and it holds under the assumption that capacitor charges and memristor
fluxes are a complete set of state variables in the FCD;

(b) we show that, via the extended mixed potential, the SEs in the FCD assume
a particularly compact and elegant form. Moreover, for some special yet
relevant subclasses of RLCM circuits, we highlight how
the mixed potential is related to the invariants of motion and invariant
manifolds that characterize the dynamics in the VCD;

(c) we discuss the relation between the mixed potential and
a fundamental concept in circuit theory, i.e., the \emph{reciprocity properties} of the
multi-port networks used in the analysis of the RLCM circuits;

(d) we illustrate the obtained results by explicitly finding the mixed
potential of basic nonlinear RLCM circuits and also high-dimensional RLCM
circuits with a neural-like architecture.

It is worth to stress that the use of FCAM plays a fundamental role
to show the existence of a mixed potential.
Indeed, a crucial step in the analysis is an equivalent principle established
in the paper via FCAM between the considered class of nonlinear RLCM circuits
in the FCD and nonlinear RLC circuits in the VCD.

This paper is mainly devoted to the theoretical aspects and the basic properties
of the mixed potential for the considered class of RLCM circuits.
In a companion paper \cite{DiMarco2026BraytonAppl}, we investigate on more practical
yet relevant aspects implied by the existence of a mixed potential. In particular,
we will address how to systematically derive convergence results for subclasses of
RLCM circuits, with the aim to extend the convergence results obtained in
\cite{di2025robust}
and other contributions in the literature \cite{HUIJZER2025106156,10141998} to memristor circuits
containing both capacitors and inductors.

The paper is organized as follows. In Sect.\ \ref{sect:class_RLCM} we introduce the class RLCM of memristor
circuits studied in the paper, while in Sect.\ \ref{sect:complete_variables} we discuss the main
hypothesis of a complete set of circuit variables to describe RLCM circuits. The SEs describing the dynamics
are derived in Section\ \ref{sect:SEs}, while Sect.\ \ref{sect:P for RLCM} introduces the mixed
potential for RLCM circuits. Two special classes of RLCM circuits are introduced in
Sect.\ \ref{sect:special}. Section\ \ref{sect:disc} discusses the significance of the obtained
results and Sect.\ \ref{sect:examples} illustrates
the main results with a number of selected
examples. Finally, the main conclusions are collected in
Sect.\ \ref{sect:concl}.

\ \vskip1.cm

\section{Class RLCM of Memristor Circuits}
\label{sect:class_RLCM}

\subsection{Representation in the VCD}

In the paper, we consider a class of nonlinear circuits $\mathfrak{N}$, denoted RLCM, containing the
following elements:

\begin{itemize}
\item[$\bullet$] $n_R$ ideal resistors $v_{Ri}=R_i i_{Ri}$;
  \item[$\bullet$] $n_e$ independent voltage sources $v_i(t)=e_{i}(t)$ and $n_a$ independent current sources $i_i(t)=a_i(t)$;
  \item[$\bullet$] $n_C$ capacitors
\begin{equation}\label{Ci}
    i_{Ci}=C_i \dot v_{Ci}
\end{equation}
and $n_L$ inductors
\begin{equation}\label{Li}
    v_{Li}=L_i \dot i_{Li};
\end{equation}
    \item[$\bullet$] $n_{\Phi}$ ideal flux-controlled memristors, denoted $M_{\Phi i}$, which are defined according to L. Chua \cite{Chua1971} by the constitutive relation (CR)
  $$
  Q_{Mi}= \hat Q_{Mi}(\Phi_{Mi})
  $$
  where
  $\Phi_{Mi}(t)=\int_{-\infty}^t v_{Mi}(\sigma)d\sigma$ is the \emph{flux} (or voltage momentum) and
  $Q_{Mi}(t)=\int_{-\infty}^t i_{Mi}(\sigma)d\sigma$ is the \emph{charge} (or current momentum). Moreover,
  $\hat Q_{Mi}: \R \to \R$ is the nonlinear memristor characteristic. We suppose henceforth that
  $\hat Q_{Mi}(0)=0$ and $\hat Q_{Mi} \in C^1(\R)$. Differentiating in time we
  obtain that the memristor satisfies in the traditional VCD the following CR given by a differential algebraic equation
  \begin{equation}\label{FCmemVCD}
    \left\{
      \begin{array}{ll}
        i_i= \hat Q_{Mi}'(\Phi_{Mi})v_i\\
        \dot \Phi_{Mi}=v_i
      \end{array}
    \right.
  \end{equation}
  where $\hat Q_{Mi}'(\Phi_{Mi})$, whose dimension is Ohm$^{-1}$, is termed state-dependent memductance;

  \item[$\bullet$] $n_{Q}$ ideal charge-controlled memristors, denoted $M_{Qi}$, with CRs $\Phi_{Mi}= \hat \Phi_{Mi}(Q_{Mi})$, where
   $\hat \Phi_{Mi}: \R \to \R$ is the nonlinear memristor characteristic such
   that $\hat \Phi_{Mi}(0)=0$ and $\hat \Phi_{Mi} \in C^1(\R)$. Differentiating
in time we
  obtain that the memristor satisfies in the VCD the following CR
  \begin{equation}\label{FCmemVCDchar}
    \left\{
      \begin{array}{ll}
        v_i= \hat \Phi_{Mi}'(Q_{Mi})i_i\\
        \dot Q_{Mi}=i_i
      \end{array}
    \right.
  \end{equation}
  where $\hat Q_{Mi}'(\Phi_{Mi})$, whose dimension is Ohm, is termed state-dependent memristance.

\end{itemize}

\subsection{Representation in the FCD}
\label{sect:FCD}


In the paper, a central role is played by the
Flux-Charge Analysis Method (FCAM) developed in \cite{Corinto-Forti-I,cfc2020}. This is based on using as electric quantities of each circuit element the \emph{incremental flux}
$ \varphi(t)=\int_{0}^t v(\sigma)d\sigma$ and \emph{incremental charge}
$q(t)=\int_0^t i(\sigma)d\sigma$, where we assume $t=0$ as the initial instant.
A circuit $\mathfrak{N}$ in the class RLCM can be analyzed
in the FCD by means of the CRs of circuit elements,
Kirchhoff-flux-law (K$\varphi$L) and Kirchhoff-charge-law (K$q$L) for incremental quantities.

The CRs of circuit elements in terms of incremental quantities are described in \cite{Corinto-Forti-I} and they are reported for convenience below.\footnote{Henceforth, for simplicity we omit the adjective incremental whenever there is no ambiguity.}

\begin{figure}[t]
  \centering
\includegraphics[width=0.8\linewidth]{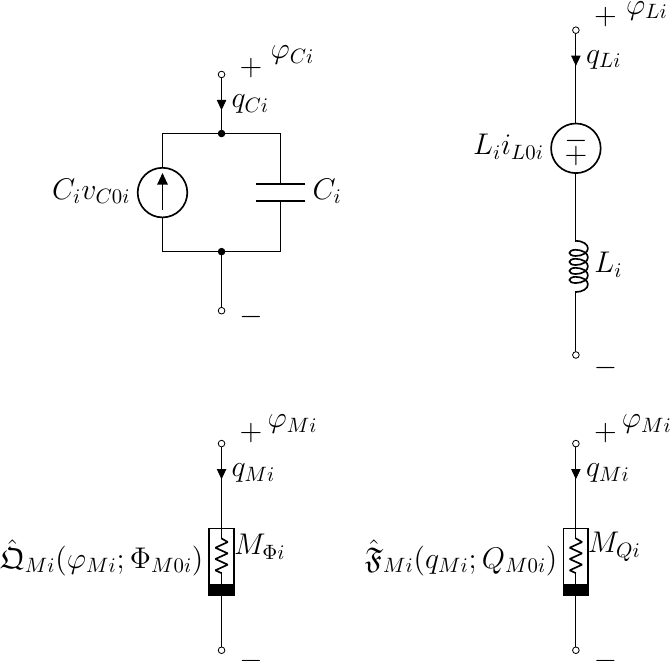}
\caption{\small Equivalent circuit of a capacitor (upper left plot) and an inductor (upper
right plot) in the FCD. A flux-controlled memristor in the FCD is the equivalent of an initial-condition
dependent voltage-controlled nonlinear resistor in the VCD (lower left plot), while a
charge-controlled memristor in the FCD is the equivalent of an initial-condition
dependent current-controlled nonlinear resistor in the VCD (lower right plot).}
\label{fig:eleFCD}
\end{figure}

\begin{itemize}

\item
A resistor is described by $\varphi_{Ri}=R_i  q_{Ri}$;

\item
an independent current (resp., voltage) source  satisfies $q_i(t)=\int_0^t a_i(\sigma)d\sigma$ (resp., $ \varphi_i(t)=\int_0^t e_i(\sigma)d\sigma$);

\item
a capacitor is described by
\begin{equation}\label{Ci FCD}
    C_i \frac{d \varphi_{Ci}}{dt}= q_{Ci}+C_iv_{C0i}
\end{equation}
where $v_{C0i}=v_{Ci}(0)$ is the initial capacitor voltage.
Its equivalent circuit in the FCD is shown in Fig.\ \ref{fig:eleFCD}.

\item
An inductor is described by
\begin{equation}\label{Li FCD}
    L_i \frac{d q_{Li}}{dt}= \varphi_{Li}+L_i i_{L0i}
\end{equation}
where $i_{L0i}=i_{Li}(0)$ is the initial inductor current.
Its equivalent circuit in the FCD is shown in Fig.\ \ref{fig:eleFCD}.

\item
A flux-controlled memristor $M_{\Phi i}$ satisfies
\begin{align}\label{Mi FCD}
\begin{split}
     q_{Mi}=&\hat{\mathfrak{Q}}_{Mi}(\varphi_{Mi};\Phi_{M0i})\\
 \doteq & \hat Q_{Mi}( \varphi_{Mi}+\Phi_{M0i})-\hat Q_{Mi}(\Phi_{M0i})
\end{split}
\end{align}
where $\Phi_{M0i}=\Phi_{Mi}(0)$ is the initial memristor flux, while a charge-controlled
memristor $M_{Qi}$ is defined by
\begin{align}\label{Mi FCD CC}
\begin{split}
     \varphi_{Mi}=&\hat{\mathfrak{F}}_{Mi}(q_{Mi};Q_{M0i})\\
     \doteq &
     \hat \Phi_{Mi}( q_{Mi}+Q_{M0i})-\hat \Phi_{Mi}(Q_{M0i})
\end{split}
\end{align}
where $Q_{M0i}=Q_{Mi}(0)$ is the initial memristor charge.

\end{itemize}

\subsection{Fundamental Equivalence Principle}
\label{sect:equiv}
Let us consider the analogies $\varphi \sim v$ and $ q \sim i$.
By comparing the CRs (\ref{Ci}) and (\ref{Ci FCD}), it is seen that a capacitor in the
VCD is the analogous of a capacitor in the FCD. Note, however, that its equivalent circuit
in the FCD (Fig.\ \ref{fig:eleFCD}) has an independent charge source accounting for the initial condition $v_{C0i}$ at $t=0$. Similar analogies hold for an inductor, a resistor and
for the two types of independent sources.

A flux-controlled memristor is described in the FCD by an algebraic
relationship $\hat{\mathfrak{Q}}_{Mi}(\cdot;\Phi_{M0i})$ between flux and charge (cf.\ (\ref{Mi FCD})).
This means that via the considered analogies a flux-controlled memristor in the FCD is the equivalent
of a nonlinear voltage-controlled resistor in the VCD. Clearly, the memristor is not equivalent to a nonlinear resistor in the VCD, since it satisfies the differential algebraic equation
(\ref{FCmemVCD}).
Another observation is that the nonlinear characteristic $\hat{\mathfrak{Q}}_{Mi}(\cdot;\Phi_{M0i})$
depends upon the initial memristor state $\Phi_{M0i}$. In conclusion, a flux-controlled memristor
in the FCD is the equivalent of an \emph{initial-condition-dependent} nonlinear voltage-controlled resistor
with characteristic $\hat{\mathfrak{Q}}_{Mi} (\cdot; \Phi_{M0i})$ in the VCD.
Analogous considerations hold for a charge-controlled memristor. See Fig.\ \ref{fig:eleFCD}.


On the basis of the previous analogies, a circuit $\mathfrak{N}$ in the class RLCM is the analogous in the FCD of a nonlinear RLC circuit in the VCD.
This simple yet important observation implies that we can use existing results in the literature for nonlinear RLC circuits to study the dynamical behavior of memristor circuits in the class RLCM.
When relying on this analogy, we should take into account that the CRs of the capacitors,
inductors and memristors in the FCD bring with them the information on the corresponding initial
conditions of the state variable in the VCD.

\section{Complete Set of Variables}
\label{sect:complete_variables}

In this section, we discuss some basic assumptions we will use for writing an SE representation of an RLCM circuit.
According to the results in \cite{chua1980}, we start by considering the following hypotheses.

\begin{assumption}
\label{assu:topo assu RLCM} Given a memristor circuit $\mathfrak{N}$ in the class RLCM, we suppose that:

 a) There is no loop made exclusively of capacitors and/or independent voltage sources and no cut-set made exclusively of inductors and/or independent current sources.

 b) Each flux-controlled memristor is in parallel to a capacitor and each charge-controlled memristor is in series with an inductor.

\end{assumption}

Condition\ a) is a standard topological assumption made in circuit theory to ensure that we can use capacitor voltages and inductor currents as independent variables to write SEs in the VCD. They also ensure we can use capacitor fluxes and inductor charges as independent variables to write SEs in the FCD. If needed, the unwanted loops and cut-sets can be removed via the techniques discussed in \cite[Sect.\ III]{chua1980} (details are omitted). Condition\ b) is once more a typical condition used to write the SEs (cf.\ \cite[Th.\ 2]{chua1980}).

\begin{figure}[t]
\begin{center}
\begin{subfigure}{0.7\columnwidth}
  \includegraphics[width=\linewidth]{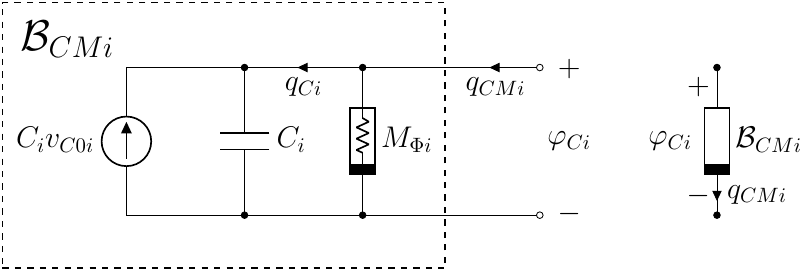}
  \caption{\small \small }
\end{subfigure}\\
\begin{subfigure}{0.7\columnwidth}
  \includegraphics[width=\linewidth]{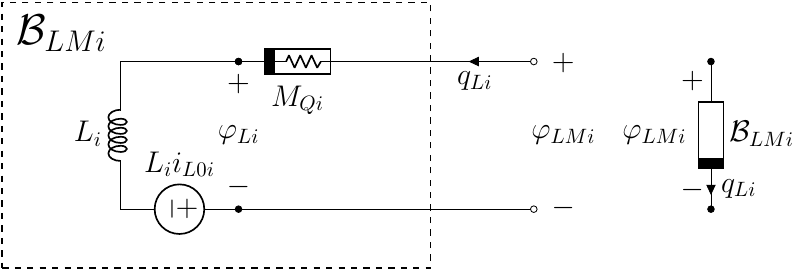}
  \caption{\small \small }
\end{subfigure}
\caption{\small (a) Representation in the FCD of a composite element $\cB_{CMi}$ given by
a capacitor in parallel to a flux-controlled memristor and (b) a composite
element $\cB_{LMi}$ given by an inductor in series with
a charge-controlled memristor.}
\label{fig:composite}
\end{center}
\end{figure}

Consider now the equivalent circuit of $\mathfrak{N}$ in the FCD.
Due to Assumption\ \ref{assu:topo assu RLCM}(b), we suppose henceforth that $\mathfrak{N}$
contains the following two terminal elements:

(a) $n_\Phi$ composite two-terminal elements, which are termed $\cB_{CMi}$, $i=1,\dots,n_\Phi$, where $0 \leq n_\phi \le n_C$, given by a capacitor
in parallel to a flux-controlled memristor (cf.\ Fig.\ \ref{fig:composite}(a)) and $n_C-n_\Phi$ elements $\cB_{Ci}$,
$i=n_{\Phi}+1,\dots,n_C$,
given by a capacitor without a memristor in parallel;

(b) $n_Q$ composite two-terminal elements, which are termed $\cB_{LMi}$, $i=1,\dots,n_Q$, where $0 \leq n_Q \le n_L$, given by an inductor in series with a charge-controlled memristor (cf.\ Fig.\ \ref{fig:composite}(b)) and $n_L-n_Q$ elements $\cB_{Li}$,
$i=n_Q+1,\dots,n_L$, given by an inductor without a memristor
in series;

(c) $n_R$ resistors, $n_e$ independent flux sources and $n_a$ independent charge sources.

From Fig.\ \ref{fig:composite}, each element $\cB_{CMi}$ satisfies the CR ($i=1,\dots,n_\Phi$)
$$
q_{CMi}=C_i \dot \varphi_{Ci}+\hat Q_{Mi}(\varphi_{Ci}+\Phi_{M0i})-\hat Q_{Mi}(\Phi_{M0i})-C_iv_{C0i}.
$$
In view of the analysis that follows, we find it convenient to put together the elements in point (a) and use vector notations to describe their CR. To this end, let

\begin{itemize}

\item[$\bullet$] $\varphi_C=(\varphi_{Ci})_{i=1,\dots,n_C} \in \R^{n_C}$;

\item[$\bullet$] $q=((q_{CMi})_{i=1,\dots,n_\Phi}, (q_{Ci})_{i=n_\Phi+1,\dots,n_C}) \in \R^{n_C}$;

\item[$\bullet$] $\varphi_{CM}=(\varphi_{Ci})_{i=1, \dots,n_\Phi} \in \R^{n_\Phi}$;

\item[$\bullet$] $\hat Q_{MC}(\cdot)=(\hat{Q}_{Mi}(\cdot))_{i=1, \dots,n_\Phi}: \R^{n_\Phi} \to \R^{n_\Phi}$;

\item[$\bullet$] $\Phi_{M0}=(\Phi_{M0i})_{i=1, \dots,n_\Phi} \in \R^{n_\Phi}$;

\item[$\bullet$] $v_{C0}=(v_{C0i})_{i=1,\dots,n_C} \in \R^{n_C}$;

\item[$\bullet$] $C={\rm diag }(C_1,\dots,C_{nC}) \in \R^{n_C \times n_C}$.
\end{itemize}

Then, the CR can be compactly written as
\begin{equation}\label{CR CM C}
    q
    =C \dot
     \varphi_C
    +\begin{pmatrix}
       \hat Q_{MC}( \varphi_{CM}+\Phi_{M0})-
       \hat Q_{MC}(\Phi_{M0})\\
       0 \\
     \end{pmatrix}
     -C v_{C0}.
\end{equation}

Similarly, from Fig.\ \ref{fig:composite}, each element $\cB_{LMi}$ satisfies the CR ($i=1,\dots,n_Q$)
$$
\varphi_{LMi}=L_i \dot q_{Li}+\hat \Phi_{Mi}(q_{Li}+Q_{M0i})-\hat \Phi_{Mi}(Q_{M0i})-L_ii_{L0i}.
$$
To use vector notations, let

\begin{itemize}

\item[$\bullet$] $q_L=(q_{Li})_{i=1,\dots,n_L} \in \R^{n_L}$;

\item[$\bullet$] $\varphi=((\varphi_{LMi})_{i=1,\dots,n_Q}, (\varphi_{Li})_{i=n_Q+1,\dots,n_L}) \in \R^{n_L}$;

\item[$\bullet$] $q_{LM}=(q_{Li})_{i=1, \dots,n_Q} \in \R^{n_Q}$;

\item[$\bullet$] $\hat \Phi_{ML}(\cdot)=(\hat \Phi_{Mi}(\cdot))_{i=1, \dots,n_Q}: \R^{n_Q} \to \R^{n_Q}$;

\item[$\bullet$] $Q_{M0}=(Q_{M0i})_{i=1, \dots,n_Q} \in \R^{n_Q}$;

\item[$\bullet$] $i_{L0}=(i_{L0i})_{i=1,\dots,n_L} \in \R^{n_L}$;

\item[$\bullet$] $L={\rm diag }(L_1,\dots,L_{nL}) \in \R^{n_L \times n_L}$.

\end{itemize}

Then, the CR of the element in point\ (b) can be written as
\begin{equation}\label{CR LM L}
     \varphi
    =L \dot q_L
    +\begin{pmatrix}
       \hat \Phi_{ML}( q_{LM}+Q_{M0})-
       \hat \Phi_{ML}(Q_{M0})\\
       0 \\
     \end{pmatrix}
     -L i_{L0}.
\end{equation}

To proceed, we need to introduce the concept of a \emph{complete set of electric variables} for an RLCM circuit $\mathfrak{N}$ in the FCD. Consider the analogy between an RLCM circuit in the FCD and
an RLC circuit in the VCD. Then, based on \cite{brayton1964theoryI}, we define a complete set as follows.

\begin{definition}
\label{def:complete_set}
Capacitor fluxes $\mathbb{F}_C=(\varphi_{Ci})_{i=1,\dots,n_C} \in \R^{n_C}$
and inductor charges $\mathbb{Q}_L=(q_{Li})_{i=1,\dots,n_L} \in \R^{n_L}$
are a complete set of variables for an RLCM circuit $\mathfrak{N}$ in the FCD if they can be chosen independently without violating Kirchhoff laws. Moreover, either the flux or charge of each element in $\mathfrak{N}$ can be found via $\mathbb{F}_C$ and/or $\mathbb{Q}_L$ directly from Kirchhoff laws, i.e., without using the CR of the element.
\end{definition}



Let us discuss how we can check in practice if $\mathbb{F}_C$ and $\mathbb{Q}_L$ are a complete set.
Consider $\mathfrak{N}$ and an associated directed graph (digraph). Suppose without loss of generality that the graph is connected and there are $b$ branches and $n$ nodes.
Clearly, $\mathbb{F}_C$ is a complete set of variables if there are $n-1$ capacitors and the elements $\cB_{CM}$ together with $\cB_C$ are a maximal tree $\cT$, i.e., they are a connected set spanning all nodes without forming loops \cite{chua-book}. Dually, $\mathbb{Q}_L$ is a complete set of variables if there are $b-n+1$ inductors and the elements $\cB_{LM}$ and $\cB_L$ form a co-tree
$\cL$. To check whether $\mathbb{F}_C$ together with $\mathbb{Q}_L$ form a complete set of variables, we consider the next conditions.

\begin{assumption}
\label{assu:case 2 RLCM 1}
There is a tree $\cT$ and a corresponding co-tree $\cL$ such that:

1) all elements $\cB_{CM}$ and $\cB_C$ belong to $\cT$ and the remaining elements of $\cT$ are charge-controlled (resistors and/or independent flux-sources);

2) all elements $\cB_{LM}$ and $\cB_L$ belong to $\cL$ and the remaining elements of $\cL$ are flux-controlled (resistors and/or independent charge sources);

3) each resistor and each independent charge source in $\cL$ forms a loop exclusively with elements $\cB_{CM}$ and/or $\cB_C$.
\end{assumption}

Due to the equivalence principle in Sect.\ \ref{sect:equiv}, and the discussion in
\cite[p.\ 4]{brayton1964theoryI}, see also ~\cite[p.\ 214]{chua1978complete}, it is seen that Assumption\ \ref{assu:case 2 RLCM 1}
guarantees that $\mathbb{F}_C$ and $\mathbb{Q}_L$ are a complete set of variables for
$\mathfrak{N}$.

\begin{remark}
In practice, to check if $\mathbb{F}_C$ and $\mathbb{Q}_L$ are a complete set of variables we need
to find a tree such that Assumption\ \ref{assu:case 2 RLCM 1} is satisfied. However,
often it is possible to check by inspection that this property holds by directly
looking at the circuit structure, see Sect.\ \ref{sect:examples} for specific examples.
We also note that, under Assumption\ \ref{assu:case 2 RLCM 1}, it follows that condition\ a) of
Assumption\ \ref{assu:topo assu RLCM} holds.
\end{remark}

\begin{remark}
It is easy to show that $(v_{Ci})_{i=1,\dots,n_C}$ and $(i_{Li})_{i=1,\dots,n_L}$ are a complete set for $\mathfrak{N}$ in the VCD if and only
if $\mathbb{F}_C$
and $\mathbb{Q}_L$ are a complete
set of variables for $\mathfrak{N}$ in the FCD (details are omitted).
\end{remark}

%
%

\section{State Equations}
\label{sect:SEs}

Goal of this section is to find an SE representation of a circuit $\mathfrak{N}$ in the class RLCM, first in the FCD and then in the VCD, under the assumption that $\mathbb{F}_C$ and $\mathbb{Q}_L$ are a complete set of variables.
Thereafter, in Sect.\ \ref{sect:P for RLCM}, we introduce an effective representation of the SEs in the FCD via a \emph{mixed potential}.


\subsection{Kirchhoff Laws in the FCD for a Complete Set of Variables}

Consider a memristor circuit $\mathfrak{N}$ in the class RLCM and an associated digraph.
Let $\cT$ be a tree and $\cL$ the corresponding co-tree.
Let $\varphi_\cT \in \R^{n-1}$, $q_\cT \in \R^{n-1}$ be the tree fluxes and charges, respectively. Moreover, let $\varphi_\cL \in \R^{b-n+1}$, $q_\cL \in \R^{b-n+1}$ be the co-tree fluxes and charges, respectively.
Due to classic results in circuit theory, we have the following (see, e.g., \cite{chua-book}).
K$q$L for the fundamental cut-sets can be written as
\begin{equation}\label{kQL}
    \begin{pmatrix}
      E_{n-1}, A \\
    \end{pmatrix}
    \begin{pmatrix}
      q_\cT \\
      q_\cL \\
    \end{pmatrix}
    =0
\end{equation}
where $E_{n-1}$ is the $(n-1) \times (n-1)$ identity matrix and $A \in \R^{(n-1) \times (b-n+1)}$
is a \emph{topological matrix} with elements $\{ -1,0,1 \}$.
The branch fluxes are obtained as
\begin{equation}
\label{branch v}
\begin{pmatrix}
  \varphi_\cT \\
  \varphi_\cL \\
\end{pmatrix}
=
\begin{pmatrix}
  E_{n-1} \\
  A^\top \\
\end{pmatrix}
\varphi_\cT
\end{equation}
hence
\begin{equation}\label{phiL}
    \varphi_\cL=A^\top \varphi_\cT.
\end{equation}
Moreover, K$\varphi$L for the fundamental loops are given as
\begin{equation}\label{kvL}
    \begin{pmatrix}
      -A^\top & E_{b-n+1} \\
    \end{pmatrix}
    \begin{pmatrix}
      \varphi_\cT \\
      \varphi_\cL \\
    \end{pmatrix}
    =0
\end{equation}
and for branch charges we have
\begin{equation}
\label{branch q}
\begin{pmatrix}
  q_\cT \\
  q_\cL \\
\end{pmatrix}
=
\begin{pmatrix}
  -A \\
  E_{b-n+1} \\
\end{pmatrix}
q_\cL.
\end{equation}
Hence
\begin{equation}\label{qT}
    q_\cT=-A q_\cL.
\end{equation}

Now, suppose that Assumption\ \ref{assu:case 2 RLCM 1} is satisfied.
Then, it is possible to partition vectors $\varphi_\cT$ and $q_\cT$ as follows
\begin{equation}\label{v q T partition}
    \varphi_\cT=
    \begin{pmatrix}
      \varphi_C \\
      \varphi_{R \cT} \\
      \varphi_e \\
    \end{pmatrix}; \hskip0.8cm
    q_\cT =\begin{pmatrix}
      q \\
      q_{R \cT} \\
      q_e \\
    \end{pmatrix}
\end{equation}
where $\varphi_C$ are the capacitor fluxes, while $\varphi_{R \cT}$ and $\varphi_e$ denote the fluxes of resistors and independent flux-sources belonging to $\cT$, respectively. Moreover,
$q$ are the charges of the $\cB_{MC}$ and $\cB_C$ elements, while $q_{R \cT}$ and $q_e$ denote the charges of resistors and independent charge sources belonging to $\cT$, respectively.

Similarly, partition vectors $\varphi_\cL$ and $q_\cL$ as follows
\begin{equation}\label{v q L partition}
    \varphi_\cL=
    \begin{pmatrix}
      \varphi \\
      \varphi_{R \cL} \\
      \varphi_a \\
    \end{pmatrix}; \hskip0.8cm
    q_\cL =\begin{pmatrix}
      q_L \\
      q_{R \cL} \\
      q_a \\
    \end{pmatrix}
\end{equation}
where $\varphi$ are the fluxes of the $\cB_{LM}$ and $\cB_L$ elements, while $\varphi_{R \cL}$ and $\varphi_a$ denote the fluxes of resistors and independent flux sources belonging to $\cL$, respectively. Moreover,
$q_L$ are the inductor charges, while $q_{R \cL}$ and $q_a$ denote the charges of resistors and independent charge sources belonging to $\cL$, respectively.

We can also partition $A$ accordingly as
\begin{equation}\label{P partition}
A=    \begin{pmatrix}
      A_{CL} & A_{CR} & A_{Ca} \\
      A_{RL} & A_{RR} & A_{Ra} \\
      A_{eL} & A_{eR} & A_{ea} \\
    \end{pmatrix}
\end{equation}
with each sub-matrix having suitable dimension.

\begin{prop}
\label{prop:P under assumpion 4}
Suppose that Assumptions\ \ref{assu:topo assu RLCM}, \ref{assu:case 2 RLCM 1}
are satisfied. Then, $A$ has the following form
\begin{equation}\label{P partition reduced}
A=
    \begin{pmatrix}
      A_{CL} & A_{CR} & A_{Ca} \\
      A_{RL} & 0 & 0 \\
      A_{eL} & 0 & 0 \\
    \end{pmatrix}.
\end{equation}
\end{prop}
\emph{Proof.}
Let us write $\varphi_{R \cL}$ and $\varphi_a$ using~(\ref{phiL}), with $A$ as in (\ref{P partition}):
\begin{align*}
\varphi_{R \cL} &= A^\top_{CR} \varphi_C + A^\top_{RR} \varphi_{R \cT} + A^\top_{eR} \varphi_e \\
\varphi_a &=A^\top_{Ca} \varphi_C + A^\top_{Ra} \varphi_{R \cT} + A^\top_{ea} \varphi_e.
\end{align*}
Due to point 3) of Assumption~\ref{assu:case 2 RLCM 1}, $\varphi_{R \cL}$ and $\varphi_a$ can be expressed using solely elements of $\varphi_C$.
As a consequence, submatrices $A_{RR}$, $A_{Ra}$, $A_{eR}$ and $A_{ea}$ are null.
\qed

Write K$q$L~(\ref{kQL}) for the quantities $q$ and K$\varphi$L~(\ref{kvL}) for the quantities $\varphi$
\begin{align}
\label{eq:qC phiL}
\begin{split}
q&=-A_{CL} q_{L} - A_{CR} q_{R \cL} - A_{Ca} q_a \\
\varphi& =A_{CL}^\top \varphi_{C} + A_{RL}^\top \varphi_{R_\cT}+ A_{eL}^\top \varphi_e.
\end{split}
\end{align}
Using the CRs of resistors, we have $q_{R \cL}= (R_{\cL})^{-1} \varphi_{R \cL}$ and $\varphi_{R_\cT} = R_{\cT} q_{R \cT}$, where $R_{\cL}$ (resp., $R_{\cT}$) is a diagonal matrix whose entries are the resistances of the resistors belonging to $\cL$ (resp., $\cT$).
Using~(\ref{phiL}) and~(\ref{qT}) and recalling the form of $A$ due to Proposition~\ref{prop:P under assumpion 4}, we obtain $q_{R \cT}=-A_{RL}q_{L}$ and $ \varphi_{R \cL}=A_{CR}^\top \varphi_{C}$.
Then, by straightforward substitutions, we can rewrite~(\ref{eq:qC phiL}) as
\begin{align*}
q & = -A_{CL} q_{L} - A_{CR} (R_{\cL})^{-1} A_{CR}^\top \varphi_{C}  - A_{Ca} q_a \\
\varphi & =A_{CL}^\top \varphi_{C} - A_{RL}^\top R_{\cT} A_{RL} q_{L}+ A_{eL}^\top \varphi_e.
\end{align*}

Summing up, in matrix-vector notation we have
\begin{equation}\label{H repr}
    \begin{pmatrix}
       q \\
       \varphi \\
    \end{pmatrix}
    =
    -H \begin{pmatrix}
       \varphi_C \\
       q_L\\
    \end{pmatrix}
    -B \begin{pmatrix}
     \varphi_e \\
     q_a \\
    \end{pmatrix}
\end{equation}
where $H \in \R^{(n_C+n_L) \times (n_C+n_L)}$, $B \in \R^{(n_C+n_L) \times (n_e+n_a)}$ are constant matrices
given by
\begin{equation}\label{H cas2}
    H=
      \begin{pmatrix}
        A_{CR} (R_{\cL})^{-1} A_{CR}^\top & A_{CL} \\
        -A_{CL}^\top & A_{RL}^\top R_{\cT} A_{RL} \\
      \end{pmatrix}
\end{equation}
and
\begin{equation}\label{B case2}
    B=\begin{pmatrix}
        A_{Ca} & 0 \\
        0 & -A_{eL}^\top \\
      \end{pmatrix}.
\end{equation}

\subsection{SEs in the FCD and VCD}

Using (\ref{H repr}), together with the CRs (\ref{CR CM C}) and
(\ref{CR LM L}), we obtain the SEs describing the behavior of $\mathfrak{N}$ in the FCD
\begin{align}\label{SEs FCD RLCM sour}
\begin{split}
    \begin{pmatrix}
      C \dot \varphi_C \\
      L \dot q_L \\
    \end{pmatrix}
=&
-H
\begin{pmatrix}
       \varphi_C \\
       q_L \\
    \end{pmatrix}\\
 &
        -\begin{pmatrix}
       \hat Q_{MC}( \varphi_{CM}+\Phi_{M0})-
       \hat Q_{MC}(\Phi_{M0})\\
       0 \\
       \hat \Phi_{ML}( q_{LM}+Q_{M0})-
       \hat \Phi_{ML}(Q_{M0})\\
       0 \\
     \end{pmatrix} \\
     &-B \begin{pmatrix}
     \varphi_e \\
     q_a \\
    \end{pmatrix}
     +
     \begin{pmatrix}
      C v_{C0} \\
      L i_{L0} \\
    \end{pmatrix}
\end{split}
\end{align}
where $H$ and $B$ are given in (\ref{H repr}) and (\ref{B case2}), respectively. Note that the
vector field defining the SEs depends upon the initial conditions
\begin{equation}\label{w0}
    w_0=(v_{C0}, i_{L0}, \Phi_{M0},Q_{M0}) \in \R^{n_C+n_L+n_\Phi+n_Q}
\end{equation}
for the state variables in the VCD.

Henceforth, we suppose that the sources in the VCD vanish for $t \ge 0$, i.e., we have $a_i(t)=0$ and $e_i(t)=0$, when $t \ge 0$. Actually, the sources $a_i(t)$ and $e_i(t)$ for $t \le 0$ are used to set initial conditions at $t=0$ for capacitors ($v_{C0i}$), inductors ($i_{L0i}$) and memristors ($\Phi_{M0i}, Q_{M0i}$) but thereafter the circuit evolves without sources. Therefore, also the sources in the FCD vanish, i.e., $q_{ai}(t)=0$ and $\varphi_{ei}(t)=0$ for $t \ge 0$ and the SEs (\ref{SEs FCD RLCM sour}) boil down to
\begin{align}\label{SEs FCD RLCM}
\begin{split}
    \begin{pmatrix}
      C \dot \varphi_C \\
      L \dot q_L \\
    \end{pmatrix}
=& -H
\begin{pmatrix}
       \varphi_C \\
       q_L \\
    \end{pmatrix}\\
&
        -\begin{pmatrix}
       \hat Q_{MC}( \varphi_{CM}+\Phi_{M0})-
       \hat Q_{MC}(\Phi_{M0})\\
       0 \\
       \hat \Phi_{ML}( q_{LM}+Q_{M0})-
       \hat \Phi_{ML}(Q_{M0})\\
       0 \\
     \end{pmatrix}\\
&
     +
     \begin{pmatrix}
      C v_{C0} \\
      L i_{L0} \\
    \end{pmatrix}.
\end{split}
\end{align}

Differentiating (\ref{SEs FCD RLCM}), we can also obtain the SEs of $\mathfrak{N}$ in the VCD,
i.e.,
\begin{align}\label{SEs VCD RLCM}
\begin{split}
    \begin{pmatrix}
      C \dot v_C \\
      L \dot i_L \\
    \end{pmatrix}
=&
-H
\begin{pmatrix}
       v_C \\
       i_L \\
    \end{pmatrix}
        -\begin{pmatrix}
       \hat Q'_{MC}( \Phi_{M})v_{CM}\\
       0 \\
       \hat \Phi'_{ML}( Q_{M})i_{LM}\\
       0 \\
     \end{pmatrix}
     \\
     \dot \Phi_{M}=&v_{CM}\\
     \dot Q_{M}=&i_{LM}.
\end{split}
\end{align}

\begin{property}
\label{prop:rel-sol}
Suppose that Assumptions\ \ref{assu:topo assu RLCM}, \ref{assu:case 2 RLCM 1}
are satisfied. The following statements hold:

1) Let $(v_C, i_L, \Phi_M, Q_M)$ be the solution of the SEs in the VCD~(\ref{SEs VCD RLCM}) with initial conditions $w_0$. Then, $(\varphi_C=\int_0^t v_C(\sigma) d\sigma$, $q_L=\int_0^t i_L(\sigma) d\sigma)$ is the solution of the SEs in the FCD~(\ref{SEs FCD RLCM}) with initial
conditions $\varphi_C(0)=0$, $q_L(0)=0$.

2) Conversely, let $(\varphi_C, q_L)$ be the solution of the SEs in the FCD~(\ref{SEs FCD RLCM}) with initial conditions $\varphi_C(0)=0$, $q_L(0)=0$. Then, $(v_C=\dot \varphi_C, i_L=\dot q_L, \Phi_M=\varphi_{CM}+\Phi_{M0}, Q_M=q_{LM}+Q_{M0})$ is the solution of the SEs in the
VCD~(\ref{SEs VCD RLCM}) with initial conditions $w_0$.
\end{property}

\emph{Proof.}
\begin{enumerate}
\item Let $(v_C, i_L, \Phi_M, q_M)$ be a solution of~(\ref{SEs VCD RLCM}) with initial conditions $w_0$. Integration of the last two equations on $[0,t]$ provides
\begin{align*}
\Phi_M=&\int_0^t v_{CM}(\sigma) d \sigma + \Phi_{M0}\\
Q_M=&\int_0^t i_{LM}(\sigma) d \sigma + Q_{M0}.
\end{align*}
Consequently, we can rewrite the first two equations of~(\ref{SEs VCD RLCM}) as
\begin{align*}
   \begin{pmatrix}
      C \dot v_C \\
      L \dot i_L \\
    \end{pmatrix}
=&
-H
\begin{pmatrix}
       v_C \\
       i_L \\
    \end{pmatrix} \\
&        -\begin{pmatrix}
       \hat Q'_{MC}(\int_0^t v_{CM}(\sigma) d \sigma + \Phi_{M0})v_{CM}\\
       0 \\
       \hat \Phi'_{ML}(\int_0^t i_{LM}(\sigma) d \sigma + \Phi_{M0})i_{LM}\\
       0 \\
     \end{pmatrix}.
\end{align*}
Integration on $[0,t]$ leads to
\begin{align*}
    \begin{pmatrix}
      C \dot \varphi_C \\
      L \dot q_L \\
    \end{pmatrix}
- &      \begin{pmatrix}
      C v_{C0} \\
      L i_{L0} \\
    \end{pmatrix} =
-H
\begin{pmatrix}
       \varphi_C \\
       q_L \\
    \end{pmatrix}\\
- & \begin{pmatrix}
       \hat Q_{MC}( \varphi_{CM}+\Phi_{M0})-
       \hat Q_{MC}(\Phi_{M0})\\
       0 \\
       \hat \Phi_{ML}( q_{LM}+Q_{M0})-
       \hat \Phi_{ML}(Q_{M0})\\
       0 \\
     \end{pmatrix}
\end{align*}
where we have taken into account that
\begin{align*}
\int_0^t & \hat Q'_{MC}\left (\int_0^\sigma v_{CM}(\tau) d \tau + \Phi_{M0}\right) v_{CM}(\sigma) d\sigma \\
=& \hat Q_{MC}( \varphi_{CM}+\Phi_{M0})- \hat Q_{MC}(\Phi_{M0})
\end{align*}
and
\begin{align*}
\int_0^t & \hat \Phi'_{ML}\left (\int_0^\sigma i_{LM}(\tau) d \tau + Q_{M0}\right) i_{LM}(\sigma) d\sigma \\
=& \hat \Phi_{ML}( q_{LM}+Q_{M0})- \hat \Phi_{ML}(Q_{M0}).
\end{align*}
Then, $(\varphi_C(t),q_L(t))=(\int_0^t v_C(\sigma) d\sigma, \int_0^t i_L(\sigma) d\sigma)$ is the solution of~(\ref{SEs FCD RLCM}) with initial conditions $\varphi_C(0)=0$, $q_L(0)=0$.
\item Let $(\varphi_C, q_L)$ be the solution of the SE in the FCD~(\ref{SEs FCD RLCM}) with initial conditions $\varphi_C(0)=0$, $q_L(0)=0$. Differentiating both sides of~(\ref{SEs FCD RLCM})
we obtain
\begin{align*}
    \begin{pmatrix}
      C \dot v_C \\
      L \dot i_L \\
    \end{pmatrix}
=& \!
- \! H
\begin{pmatrix}
       \dot \varphi_C \\
       \dot q_L \\
    \end{pmatrix}
       \! - \! \begin{pmatrix}
       \hat Q'_{MC}(\varphi_{CM}+ \Phi_{M0})\dot \varphi_{CM}\\
       0 \\
       \hat \Phi'_{ML}( q_{LM}+ Q_{M0})\dot q_{LM}\\
       0
     \end{pmatrix}.
\end{align*}
\end{enumerate}
Now, let $\Phi_M=\varphi_{CM} + \Phi_{M0}$ and $Q_M= q_{LM}+ Q_{M0}$. By noting that $\dot \Phi_M= \dot \varphi_{CM} = v_{CM}$,
$\dot Q_M=\dot q_{LM}=i_{LM}$, $\dot \varphi_C=v_C$ and $\dot q_L=i_L$ we can rewrite the above equations as
\begin{align*}
    \begin{pmatrix}
      C \dot v_C \\
      L \dot i_L \\
    \end{pmatrix}
=&
-H
\begin{pmatrix}
       v_C \\
       i_L \\
    \end{pmatrix}
        -\begin{pmatrix}
       \hat Q'_{MC}(\Phi_{M}) v_{CM}\\
       0 \\
       \hat \Phi'_{ML}(Q_{M}) i_{LM}\\
       0
     \end{pmatrix}.
\end{align*}
This shows that $(\dot \varphi_C , \dot q_L, \varphi_{CM}+ \Phi_{M0}, q_{LM}+Q_{M0})$ is the solution of~(\ref{SEs VCD RLCM}) with initial conditions $w_0$. \qed


\subsection{A Fundamental Reciprocity Property}

It is instructive to interpret the previous results by invoking the fundamental
concept in circuit theory of \emph{reciprocity}.
As it is typically done in circuit theory, let us
decompose the memristor circuit $\mathfrak{N}$ as shown in Fig.\ \ref{fig:multiport}, i.e., a resistive $(n_C+n_L)$-port $N$ containing $n_R$ resistors, to which $n_C+n_L$ one-port elements are connected. In particular, we have $n_{\Phi}$ elements $\cB_{CM}$, $n_C-n_{\Phi}$ elements $\cB_C$, $n_Q$ elements $\cB_{LM}$ and $n_L-n_Q$ elements $\cB_L$ (cf.\ Sect.\ \ref{sect:complete_variables}).

\begin{figure}[t]
  \centering
\includegraphics[width=0.9\linewidth]{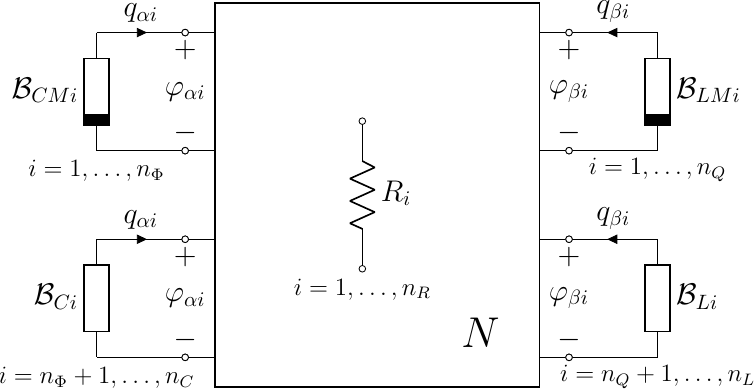}
\caption{\small Decomposition of a circuit $\mathfrak{N}$ in the class RLCM.}
\label{fig:multiport}
\end{figure}

Denote by $ \varphi_\alpha =(\varphi_{\alpha i})_{i=1,\dots,n_C}\in \R^{n_C}$ (resp., $q_\alpha =(q_{\alpha i})_{i=1,\dots,n_C} \in \R^{n_C}$) the vector of  port fluxes (resp., charges) of the capacitor ports of $N$, see Fig.\ \ref{fig:multiport}.
Moreover, denote by $ q_\beta =(q_{\beta i})_{i=1,\dots,n_L} \in \R^{n_L}$ (resp., $ \varphi_\beta =(\varphi_{\beta i})_{i=1,\dots,n_L} \in \R^{n_L}$) the vector of  port charges (resp., fluxes) of the inductor ports of $N$, see Fig.\ \ref{fig:multiport}.

\begin{prop}
\label{prop:H case2}
Suppose that Assumptions\ \ref{assu:topo assu RLCM}, \ref{assu:case 2 RLCM 1}
are satisfied. Then, there exists the hybrid-representation of the $(n_C+n_L)$-port $N$
\begin{equation}\label{H represent}
    \begin{pmatrix}
       q_\alpha \\
       \varphi_\beta \\
    \end{pmatrix}
    =
    H \begin{pmatrix}
       \varphi_\alpha \\
       q_\beta \\
    \end{pmatrix} +
    B \begin{pmatrix}
     \varphi_e \\
     q_a \\
    \end{pmatrix}
\end{equation}
where
$$
H= \begin{pmatrix}
        H_{\alpha \alpha} & H_{\alpha \beta} \\
        H_{\beta \alpha} & H_{\beta \beta} \\
      \end{pmatrix}
$$
is such that
$$
H_{\alpha \alpha}=H_{\alpha \alpha}^\top=A_{CR} (R_{\cL})^{-1} A_{CR}^\top
$$
and
$$
H_{\beta \beta}=H_{\beta \beta}^\top=A_{RL}^\top R_{\cT} A_{RL}
$$
are symmetric matrices, while
$$
H_{\alpha \beta}=A_{CL}
$$
and
$$
H_{\beta \alpha}=-H_{\alpha \beta}^\top=-A_{CL}^\top.
$$
Moreover, $B$ is given in (\ref{B case2}).
\end{prop}

\emph{Proof.} The result follows from (\ref{H repr})-(\ref{B case2}), considering that
$q_\alpha=-q$ and $\varphi_\beta=-\varphi$.
\qed

\begin{remark}
It is noticed that the particular form of $H$ in Proposition\ \ref{prop:H case2}, where $H_{\alpha \alpha}$ and  $H_{\beta \beta}$ are symmetric matrices, while we have $H_{\beta \alpha}=-H_{\alpha \beta}^\top$, is a consequence of the reciprocity of the multi-port $N$ (cf.\ \cite[Sect.\ 4.1]{chua1980}). Indeed, it is known that a multi-port network containing only two-terminal elements is necessarily reciprocal \cite[Th.\ 7]{chua1980}.
\end{remark}

\section{Brayton-Moser mixed potential}
\label{sect:P for RLCM}

In the seminal articles \cite{brayton1964theoryI,brayton1964theoryII}, Brayton and Moser proposed a theory of nonlinear RLC circuits (without memristors) based on the concept of a \emph{mixed potential}. In this section, we show that we can extend the mixed-potential to the class RLCM of circuits containing also memristors in addition to RLC components, provided we consider their description via FCAM in the FCD domain.

Consider a circuit $\mathfrak{N}$ in the class RLCM.
As seen in Sect.\ \ref{sect:equiv}, there is an analogy between $\mathfrak{N}$ in the FCD and an RLC circuit in the VCD.
Via this analogy, the dynamical behavior of $\mathfrak{N}$ can be analyzed through the mixed potential introduced by Brayton and Moser \cite{brayton1964theoryI} for nonlinear RLC circuits.

The following result holds.

\begin{theorem}
\label{th:mixed_potential}
Suppose that Assumptions\ \ref{assu:topo assu RLCM}, \ref{assu:case 2 RLCM 1}
are satisfied. Introduce for $\mathfrak{N}$ the mixed potential $\cP(\varphi_C,q_L;w_0):
\R^{(n_C+n_L)} \to \R$ defined as
\begin{align}\label{mixed n1}
\begin{split}
    \cP & (\varphi_C,q_L; w_0)=
    \frac{1}{2}  \varphi_C^\top H_{\alpha \alpha}  \varphi_C
    -\frac{1}{2}  q_L^\top H_{\beta \beta}  q_L
    +  \varphi_C^\top H_{\alpha \beta}  q_L\\
    &+\sum_{i=1}^{n_{\Phi}} \int_0^{ \varphi_{\alpha i}}(\hat Q_{Mi}(\rho+\Phi_{M0i})-
       \hat Q_{Mi}(\Phi_{M0i}))d\rho \\
    &-\sum_{j=1}^{n_{Q}} \int_0^{ q_{\beta j}}
(\hat \Phi_{Mj}(\rho+Q_{M0j})-
       \hat \Phi_{Mj}(Q_{M0j}))d\rho \\
    &- \varphi_C^\top C v_{C0}
+  q_L^\top L i_{L0}
\end{split}
\end{align}
which depends upon the initial conditions $w_0$ for the state variables in the
VCD (cf.\ (\ref{w0})).
Then, the SEs
(\ref{SEs FCD RLCM}) of $\mathfrak{N}$ in the FCD can be written in the form
\begin{equation}\label{SEs FCD RLCM mixed}
     \begin{pmatrix}
      C \dot \varphi_C \\
      L \dot q_L \\
    \end{pmatrix}
=    \begin{pmatrix}
      -\frac{\partial}{\partial  \varphi_C} \cP( \varphi_C, q_L;w_0) \\
      \frac{\partial}{\partial  q_L} \cP( \varphi_C, q_L;w_0) \\
    \end{pmatrix}.
\end{equation}
\end{theorem}

\begin{figure}[t]
  \centering
\includegraphics[width=1.0\linewidth]{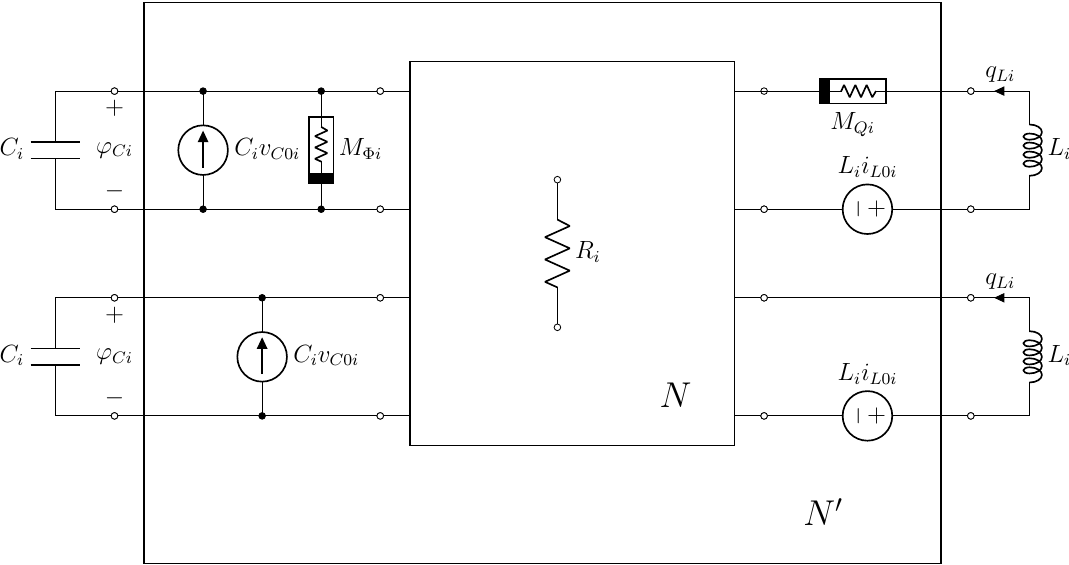}
\caption{\small Decomposition of a circuit $\mathfrak{N}$ in the class RLCM used in the proof of
Theorem\ \ref{th:mixed_potential}.}
\label{fig:N1}
\end{figure}

\emph{Proof.}
The mixed potential is a purely dc concept related to the resistive multi-port network $N'$ to which capacitors and inductors are connected (see Fig.\ \ref{fig:N1}). Next, we show that,
under Assumption\ \ref{assu:case 2 RLCM 1}, a mixed potential of $N'$ can be explicitly found via the topological matrix describing the K$q$Ls at the fundamental cut-sets (cf.\ (\ref{kQL})) and the CRs of circuit elements.
The multi-port $N'$
is constituted by the charge and flux sources related to capacitor and inductor initial conditions, respectively, the memristor nonlinearities and the resistive multi-port $N$.
Due to an additivity property \cite{brayton1964theoryI}, the mixed potential of $N'$ is given by
\begin{align*}
\begin{split}
\cP( \varphi_C, q_L;w_0)
=&\cP_0( \varphi_C, q_L; w_0)+\cP_M( \varphi_C, q_L; w_0)\\
 &+\cP_N( \varphi_C, q_L,w_0)
\end{split}
\end{align*}
where the first term is the mixed potential of the charge sources $Cv_{C0}$ and flux sources $Li_{L0}$, the second term is that of the memristors and the third is the mixed potential of the multi-port $N$.

We have
$$
\cP_0( \varphi_C, q_L;w_0)=- \varphi_C^\top C v_{C0}
+  q_L^\top L i_{L0}.
$$
The mixed potential of flux-controlled memristors in parallel to capacitors is
$$
\cP_{\Phi}( \varphi_C;w_0)= \sum_{i=1}^{n_\Phi} \int_0^{ \varphi_{C i}}(\hat Q_{Mi}(\rho+\Phi_{M0i})-
       \hat Q_{Mi}(\Phi_{M0i}))d\rho
$$
while that of the charge-controlled memristors in series with inductors is
$$
\cP_{Q}( q_L;w_0) \! = \! \! - \! \! \sum_{i=1}^{n_{Q}} \int_0^{ q_{L i}}
(\hat \Phi_{Mi}(\varrho+Q_{M0i}) \! - \!
      \hat \Phi_{Mi}(Q_{M0i}))d\varrho.
$$
Therefore, $\cP_M( \varphi_C, q_L)=\cP_{\Phi}( \varphi_C)+\cP_{Q}( q_L)$.
The mixed potential of $N$, given in terms of its $H$-representation, can be expressed as
\begin{align*}
\begin{split}
\cP_N( \varphi_C, q_L;w_0)=&
    \int_0^{\varphi_C} H_{\alpha \alpha} \sigma \cdot d \sigma \\
    &-\int_0^{q_L} H_{\beta \beta} \varsigma \cdot d \varsigma
    +\varphi_C^\top H_{\alpha \beta}  q_L
\end{split}
\end{align*}
where the dot denotes the scalar product.
On the other hand, considering that $H_{\alpha \alpha}$ is symmetric, we have
\begin{align*}
\int_0^{\varphi_C} H_{\alpha \alpha} \sigma \cdot d \sigma
=& \int_0^{\varphi_C} \frac{\partial}{\partial \sigma}
(\frac{1}{2} \sigma^\top H_{\alpha \alpha} \sigma) \cdot d\sigma \\
=&\frac{1}{2}\varphi_C^ \top H_{\alpha \alpha} \varphi_C.
\end{align*}
Similarly, since $H_{\beta \beta}$ is symmetric
\begin{align*}
\int_0^{q_L} H_{\beta \beta} \varsigma \cdot d \varsigma
=\frac{1}{2}q_L^\top H_{\beta \beta} q_L
\end{align*}
and then
\begin{align*}
\cP_N( \varphi_C, q_L;w_0)\! \! = \! \!
    \frac{1}{2}\varphi_C^ \top H_{\alpha \alpha} \varphi_C
    \! - \! \frac{1}{2}q_L^\top H_{\beta \beta} q_L
    \!+\! \varphi_C^\top H_{\alpha \beta}  q_L.
\end{align*}

Summing up these terms, we obtain the expression of the mixed potential
given in (\ref{mixed n1}).

The second part of the theorem follows by verifying (\ref{SEs FCD RLCM mixed})
via direct computation of the partial derivatives of the mixed potential.
\qed

\section{Two Special Classes of RLCM Networks}
\label{sect:special}
Next, we study in more detail two special classes of memristor circuits in the class RLCM.
In the first one, we consider the maximum possible number of memristors (Sect.\ \ref{sect:memr}),
while in the second we assume there there is only one type of reactive elements
(Sect.\ \ref{sect:RCM}).

\subsection{Maximum Number of Memristors}
\label{sect:memr}
Consider the special yet relevant case where each capacitor has in parallel a flux-controlled
memristor and each inductor has in series a charge-controlled memristor.
Now, we have $n_{\Phi}=n_C$, $n_{Q}=n_L$, while $\Phi_{MC}=\varphi_{CM}+\Phi_{MC0}$
and $Q_{ML}=q_{LM}+Q_{ML0}$.

Consider functions $\cQ:\R^{2(n_C+n_L)} \to \R^{n_C}$ and $\Psi: \R^{2(n_C+n_L)} \to \R^{n_L}$
defined as follows
\begin{align}
\begin{split}
\label{Q(t) case2}
    \cQ =&\cQ(v_C,i_L,\Phi_{MC},Q_{ML})=Cv_C+H_{\alpha \alpha} \Phi_{MC}\\
    &+H_{\alpha \beta} Q_{ML}
    +\hat Q_{MC}(\Phi_{MC})
\end{split}
\end{align}
and
\begin{align}
\begin{split}
\label{Phi(t) case2}
    \Psi =&\Psi(v_C,i_L,\Phi_{MC},Q_{ML})
    =Li_L-H_{\alpha \beta}^\top \Phi_{MC}\\
    &+H_{\beta \beta} Q_{ML}
    +\hat \Phi_{ML}(Q_{ML}).
\end{split}
\end{align}
Then, the SEs (\ref{SEs FCD RLCM}) in the FCD can be put into the form
\begin{align}\label{SEs FCD RLCM special}
\begin{split}
       \begin{pmatrix}
      C \dot \Phi_{MC}  \\
      L \dot Q_{ML} \\
    \end{pmatrix}
=&
-H
\begin{pmatrix}
       \Phi_{MC} \\
       Q_{ML} \\
    \end{pmatrix}
        -\begin{pmatrix}
       \hat Q_{MC}( \Phi_{MC} )\\
       \hat \Phi_{ML}( Q_{ML})\\
     \end{pmatrix}
    +
\begin{pmatrix}
       Q_0\\
       \Psi_0 \\
    \end{pmatrix}
\end{split}
\end{align}
where $Q_0=\cQ(0)$ is given by
\begin{align}
\begin{split}
\label{Q0 case2}
Q_0 \! = \! Cv_{C0}+H_{\alpha \alpha} \Phi_{MC0} \! + \! H_{\alpha \beta} Q_{ML0}
    \! + \! \hat Q_{MC}(\Phi_{MC0})
\end{split}
\end{align}
and $\Psi_0 = \Psi(0)$ amounts to
\begin{align}
\begin{split}
\label{Ph0i case2}
\Psi_0 \! = Li_{L0}-H_{\alpha \beta}^\top \Phi_{MC0}+H_{\beta \beta} Q_{ML0}
    +\hat \Phi_{ML}(Q_{ML0}).
\end{split}
\end{align}
These are constant terms depending upon the initial conditions for the state variables in the VCD.

The SEs in the VCD become
\begin{align}\label{SEs VCD RLCM special}
\begin{split}
    \frac{d}{dt}
    \begin{pmatrix}
      C v_C \\
      L i_L \\
    \end{pmatrix}
=&
-H
\begin{pmatrix}
       v_C \\
       i_L \\
    \end{pmatrix}
        -\begin{pmatrix}
       \hat Q'_{MC}(\Phi_{MC})v_C \\
       \hat \Phi'_{ML}(Q_{ML})i_L \\
     \end{pmatrix}
     \\
     \frac{d}{dt} \Phi_{MC}=&v_C\\
     \frac{d}{dt} Q_{ML}=&i_L.
\end{split}
\end{align}

It is immediate to check that we have
$$
\frac{d}{dt}\cQ=0; \ \ \frac{d}{dt}\Psi=0
$$
along the solutions of the SEs (\ref{SEs VCD RLCM special}) in the VCD.
This means that these SEs admit $n_C+n_L$ \emph{invariants of motion}
(or conserved quantities) given by the components of functions $\cQ$ and $\Psi$.
As a consequence, the dynamics evolves on invariant manifolds
\begin{align}
\begin{split}
\cM(Q_0,\Psi_0)=& \{ (v_C,q_L,\Phi_{MC},Q_{ML}) \in \R^{2(n_C+n_L)}:\\
 &\cQ=Q_0, \Psi=\Psi_0 \}
\end{split}
\end{align}
where $(Q_0,\Psi_0) \in \R^{n_C+n_L}$ is termed \emph{manifold index}.

Let us now consider the modified mixed potential $\tilde P(\Phi_M,Q_M; Q_0,\Psi_0): \R^{2(n_C+n_L)}\to \R$ given by
\begin{equation}\label{tildP}
\begin{split}
  \tilde \cP&(\Phi_M,Q_M; \Psi_0,Q_0)=\frac{1}{2}\Phi_M^\top H_{\alpha \alpha} \Phi_M - \frac{1}{2}Q_M^\top H_{\beta \beta} Q_M \\
  &+ \Phi_M^\top H_{\alpha \beta} Q_M + \sum_{i=1}^{n_C} \int_0^{\varphi_{Mi}} \hat{Q}_{Mi}(\sigma) d \sigma\\
  &- \sum_{j=1}^{n_L} \int_0^{q_{Mj}} \hat{\Phi}_{Mj}(\sigma) d \sigma
  + (\Phi_M^\top \, Q_M^\top) \begin{pmatrix}
 - Q_0 \\
  \Psi_0
  \end{pmatrix}.
  \end{split}
\end{equation}
Clearly, $\tilde \cP$ is manifold dependent and it is function of the manifold index $(Q_0,\Psi_0)$.
It can be verified that the SEs (\ref{SEs FCD RLCM special}) in the FCD can be written
as follows
\begin{align}\label{SEs FCD RLCM special_mixed}
\begin{split}
      \begin{pmatrix}
      C \dot \Phi_{MC}  \\
      L \dot Q_{ML} \\
    \end{pmatrix}
=&
\begin{pmatrix}
       -\frac{\partial}{\partial \Phi_M} \tilde \cP (\Phi_M,Q_M; Q_0,\Psi_0) \\
       \frac{\partial}{\partial Q_M} \tilde \cP(\Phi_M,Q_M; Q_0,\Psi_0)
    \end{pmatrix}.
    \end{split}
\end{align}

By means of $\tilde \cP$ we can also give a compact expression to the invariants of motion.
We have
\begin{equation}\label{cQmixed}
    \cQ = Cv_C+ \frac{\partial}{\partial \Phi_M} \tilde \cP(\Phi_M,Q_L;Q_0,\Psi_0)+Q_0
\end{equation}
while
\begin{equation}\label{cPsimixed}
    \Psi = Li_L - \frac{\partial}{\partial Q_M} \tilde \cP(\Phi_M,Q_L;Q_0,\Psi_0)+\Psi_0.
\end{equation}

\subsection{RCM and RLM Networks}
\label{sect:RCM}

Consider the special case where a circuit in $\mathfrak{N}$
has no inductors, i.e., $n_L=0$ (class RCM).
Under Assumptions\ \ref{assu:topo assu RLCM}, \ref{assu:case 2 RLCM 1}, the tree $\cT$ is made by all the elements $\cB_{CM}$ and $\cB_C$, while all the resistors are on the co-tree $\cL$. In this case, the SEs
in the FCD boil down to
\begin{align}
\begin{split}
\label{SEs FCD RCM}
      C \dot \varphi_C
=&-G \varphi_C
        -\begin{pmatrix}
       \hat Q_{MC}( \varphi_{CM}+\Phi_{M0})-
       \hat Q_{MC}(\Phi_{M0})\\
       0 \\
     \end{pmatrix}\\
     &+ C v_{C0}
\end{split}
\end{align}
where $G$ is a symmetric matrix of conductances given by
$$
G=A_{CR} (R_{\cL})^{-1} A_{CR}^\top.
$$
Moreover, on the basis of (\ref{mixed n1}), the mixed potential $\cP_\mathrm{RCM}(v_C;w_0):\R^{n_C} \to \R$ becomes
\begin{align}\label{mixed RC}
\begin{split}
    \cP_\mathrm{RCM}(v_C;w_0)=& \frac{1}{2}  \varphi_C^\top G  \varphi_C
    +\sum_{i=1}^{n_{\Phi}} \int_0^{ \varphi_{\alpha i}}(\hat Q_{Mi}(\rho+\Phi_{M0i})\\
    &  -\hat Q_{Mi}(\Phi_{M0i}))d\rho - \varphi_C^\top C v_{C0}
\end{split}
\end{align}
where $w_0=(v_{C0},\Phi_{M0})$.

Due to Theorem\ \ref{th:mixed_potential}, the SEs
(\ref{SEs FCD RLCM}) of $\mathfrak{N}$ in the FCD can be written in the
\emph{gradient form}
\begin{equation}\label{SEs FCD RCM mixed}
      C \dot \varphi_C =
      -\frac{\partial}{\partial \varphi_C} \cP_\mathrm{RCM}( \varphi_C;w_0).
\end{equation}
This implies that the mixed potential is actually a Lyapunov function that
decreases along the solutions of the SEs (\ref{SEs FCD RLCM}). An analogous
treatment holds when the memristor circuit contains inductors but
there are no capacitors (class RLM).

\section{Discussion}
\label{sect:disc}
Here, we collect some remarks where we discuss the significance of the previous results.

\begin{remark}
The mixed potential (\ref{mixed n1}) obtained in Theorem\ \ref{th:mixed_potential} generalizes to
RLCM circuits the mixed potential originally obtained by
Brayton and Moser \cite{brayton1964theoryI,brayton1964theoryII} for memristor-less
RLC circuits. The extension has been possible by analyzing memristor circuits in
the FCD, rather than in the VCD, as it was done in the quoted papers by Brayton and Moser.
The use of FCAM and the equivalence principle obtained in Sect.\ \ref{sect:equiv} between
an RLCM circuit in the FCD and an RLC circuit in the VCD are the key elements enabling
the extension to circuits containing memristors.
\end{remark}

\begin{remark}
The mixed potential (\ref{mixed n1}) depends upon the initial conditions
for the state variables $w_0$ in the VCD (cf.\ (\ref{w0})). It is worth to note that, for special classes
of circuits, as the class considered in Sect.\ \ref{sect:memr}, it is directly related to the
invariants of motion and the invariant manifold index $(Q_0,\Psi_0)$ (cf.\ (\ref{SEs FCD RLCM special_mixed}), (\ref{cQmixed}) and (\ref{cPsimixed})).
\end{remark}

\begin{remark}
The expression of the mixed potential (\ref{mixed n1})
will play a crucial role for the analysis of convergence of circuits in the class RLCM,
a topic which we address in a systematic way in a companion paper \cite{DiMarco2026BraytonAppl}.
\end{remark}

\begin{remark}
In Sect.\ \ref{sect:RCM}, we have seen that in the special case of class RCM (or RLM), a memristor circuit
in the FCD is described by a gradient system with respect to the mixed
potential, i.e., the mixed potential is actually a Lyapunov function.
This result is analogous to that obtained in \cite[Th.\ 1]{di2025robust}.
However, while the Lyapunov function used in that paper has been been found via an
abstract mathematical procedure, here we have shown that the Lyapunov
function has a clear circuit theoretic meaning related to the mixed
potential of the memristor circuit.

\end{remark}

\section{Illustrative Examples}
\label{sect:examples}
In this section, we provide a number of examples to show how we can explicitly find
the mixed potential for memristor circuits in the class RLCM. We start with some
basic RLCM circuits (Sect.\ \ref{sect:basic}) and then we consider large RLCM
circuits with a neural network structure (Sect.\ \ref{sect:neur}). All the
circuits are represented in the FCD in order to use the theoretic results in the
paper to write the SEs and find a mixed potential.

\subsection{Basic RLCM Circuits}
\label{sect:basic}
\subsubsection{Simple RCM Circuit}

\begin{figure}[t]
  \centering
\includegraphics[width=0.7\linewidth]{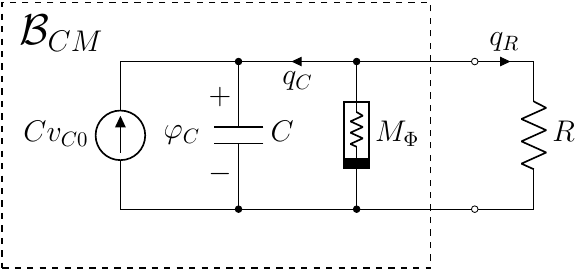}
\caption{\small Simplest memristor circuit in the class RCM.}
\label{fig:MC}
\end{figure}

Consider the simplest circuit in the class RCM (cf.\ Sect.\ \ref{sect:RCM}) given by a resistor,
a capacitor and a flux-controlled memristor in parallel with each other
(Fig.\ \ref{fig:MC}). Clearly, Assumption\ \ref{assu:case 2 RLCM 1} is satisfied.
The SE describing the circuit can be obtained by using
K$q$L, i.e., $q_C+q_M+q_R=0$, and K$\varphi$L, i.e., $\varphi_C=\varphi_M$,
together with the CRs of the capacitor, memristor
and resistor (cf.\ Sect.\ \ref{sect:FCD}). We obtain the first-order SE
$$
\dot \varphi_C=-\frac{\varphi_C}{R}+\hat Q_M(\varphi_C+\Phi_{M0})-
\hat Q_M(\Phi_{M0}) +Cv_{C0}.
$$
In this case we simply have $G=1/R$ (cf.\ (\ref{SEs FCD RCM})). According to
(\ref{mixed RC}), the mixed potential is given by
\begin{align*}
\begin{split}
\cP(\varphi_C; v_{C0},\Phi_{M0})=&\frac{1}{2R}\varphi_C^2
+\int_0^{\varphi_C} \hat Q_M(\rho+\Phi_{M0})d\rho\\
 &-\hat Q_M(\Phi_{M0})\varphi_C-Cv_{C0} \varphi_C.
\end{split}
\end{align*}
It can be checked that the circuit obeys the gradient system
$$
C \dot \varphi_C = -\frac{\partial}{\partial \varphi_C} \cP
(\varphi_C; v_{C0},\Phi_{M0}).
$$

\begin{figure}[t]
  \centering
\includegraphics[width=0.87\linewidth]{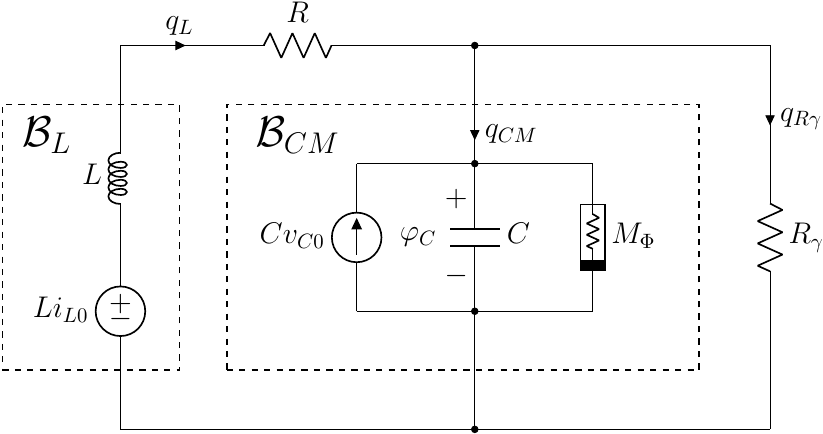}
\caption{\small Murali-Lakshmanan-Chua circuit with a memristor.}
\label{fig:mural}
\end{figure}

\subsubsection{Simple RLCM Circuit} Consider the RLCM circuit with a capacitor,
an inductor, two resistors and a flux-controlled memristor, that is depicted in Fig.\ \ref{fig:mural}.
This is a
modified version of Murali-Lakshmanan-Chua circuit~\cite{murali1994simplest} where the nonlinear resistor is replaced by a flux-controlled memristor~\cite{ishaq2013nonsmooth}. It is a core element in several neuromorphic applications, such as embedding of neuron models~\cite{Innocenti2022735}, reservoir computing~\cite{escudero2026MLCreservoire} and analogic implementation of logic gates and memory latches~\cite{ashokkumar2021MLClogic}. It can be checked
by inspection that $\varphi_C$
and $q_L$ are a complete set of variables (cf.\ Assumption\ \ref{def:complete_set}).
In fact, the charge in $R$ can be obtained from $q_L$, since these
elements are in series. Moreover, the flux on $R_\gamma$ can be obtained via $\varphi_C$, since these elements are in parallel.
The SEs can be obtained by writing K$q$L, i.e.,
$q_{CM}-q_L+q_{R\gamma}=0$,
and K$\varphi$L, i.e., $\varphi_L+\varphi_R+\varphi_C=0$ and using the CRs of circuit elements. We obtain that
the circuit is described by the second-order SEs
\begin{equation}\label{SEsmural}
   \left\{
      \begin{array}{ll}
        C\dot \varphi_C = -\frac{\varphi_C}{R_\gamma}+q_L-\hat Q_M(\varphi_C+\Phi_{M0})\\
        + \hat Q_M(\Phi_{M0})+Cv_{C0}\\
        L\dot q_L =-\varphi_C -Rq_L +Li_{L0}.
      \end{array}
    \right.
\end{equation}
Comparing with (\ref{SEs FCD RLCM}), we have
$$
H=\begin{pmatrix}
        H_{\alpha \alpha} & H_{\alpha \beta} \\
        H_{\beta \alpha} & H_{\beta \beta} \\
      \end{pmatrix}
 =\begin{pmatrix}
    \frac{1}{R_\gamma} & -1 \\
    1 & R \\
  \end{pmatrix}
$$
and the mixed potential (\ref{mixed n1}) takes the form
\begin{align*}
\begin{split}
\cP(\varphi_C,q_L;w_0)=&\frac{1}{2R_\gamma}\varphi_C^2
-\frac{1}{2}Rq_L^2-\varphi_C q_L\\
 &+\int_0^{\varphi_C} \hat Q_M(\rho+\Phi_{M0})d\rho\\
 &-\hat Q_M(\Phi_{M0})\varphi_C\\
 &-Cv_{C0} \varphi_C +Li_{L0}q_L
\end{split}
\end{align*}
where we let $w_0=(v_{C0},i_{L0},\Phi_{M0})$.
It can be verified that the SEs (\ref{SEsmural}) satisfy
\begin{equation}\label{SEsmuralmixed}
   \left\{
      \begin{array}{ll}
        C\dot \varphi_C = - \frac{\partial}{\partial \varphi_C}\cP(\varphi_C,q_L;w_0)\\
        L\dot q_L = \frac{\partial}{\partial q_L}\cP(\varphi_C,q_L;w_0).
      \end{array}
    \right.
\end{equation}

\begin{figure}[t]
\begin{center}
\begin{subfigure}{1.0\columnwidth}
  \includegraphics[width=\linewidth]{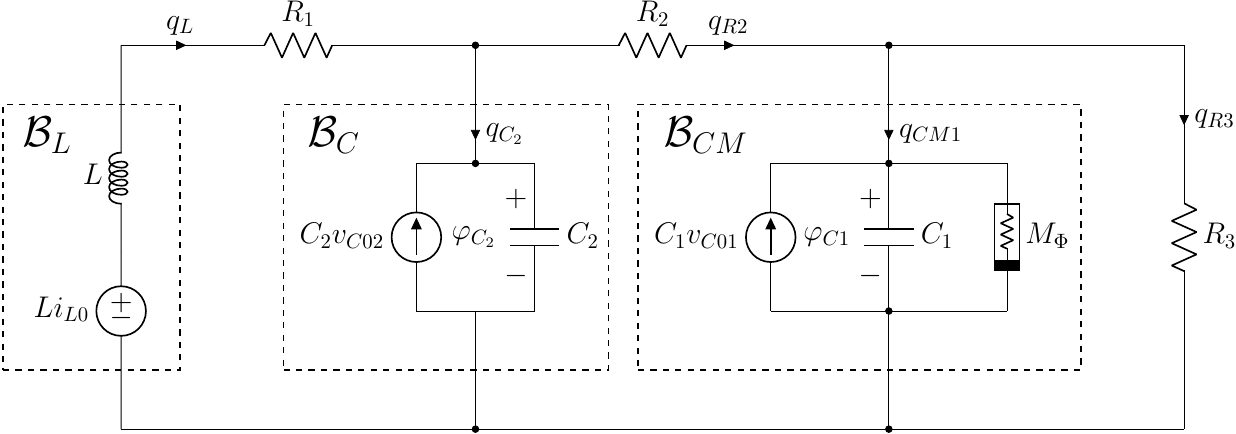}
  \caption{\small \small }
\end{subfigure}\\
\begin{subfigure}{0.4\columnwidth}
  \includegraphics[width=\linewidth]{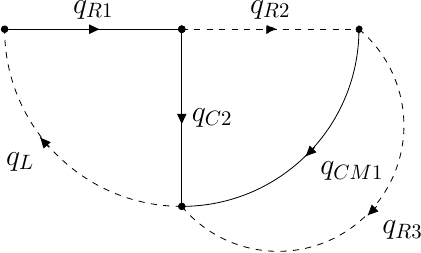}
  \caption{\small \small }
\end{subfigure}
\caption{\small (a) Modified Chua's circuit with a memristor. (b)
Associated digraph where the tree branches are indicated with solid lines and
co-tree branches with dashed lines.}
\label{fig:chua}
\end{center}
\end{figure}

\subsubsection{Memristor Chua's Circuit}
\label{sect:Chua}

Consider the RLCM circuit in Fig.\ \ref{fig:chua}(a) with two capacitors, one inductor, three
resistors and one flux-controlled memristor. This circuit, which is a modified version of the classic
Chua's circuit where the nonlinear resistor in that circuit is replaced by a flux-controlled
memristor, has been widely studied in the literature \cite{Itoh2008,rocha2017memristive,ahamed2020sliding}.
Figure\ \ref{fig:chua}(b) highlights the main elements of the circuit defined in Sect.\ \ref{sect:complete_variables}.
A digraph associated with the circuit is shown in Fig.\ \ref{fig:chua}(b). Choose a tree $\cT=\{\cB_{CM},\cB_C,R_1\}$
and a co-tree $\cL=\{\cB_{L},R_2,R_3 \}$. It is easy to see that, with these choices,
Assumption\ \ref{assu:case 2 RLCM 1} is satisfied, hence $\varphi_{C1}$, $\varphi_{C2}$
and $q_L$ are a complete set of variables for the circuit. Writing the K$q$Ls as in~(\ref{qT}), we obtain
the topological matrix $A$ and its decomposition according to~(\ref{P partition reduced})
\begin{align*}
A&= \left (\begin{array}{c|c}
A_{CL} &A_{CR} \\
\hline
A_{RL} &0
\end{array} \right) =
\left (\begin{array}{c|cc}
0&-1 &1 \\
-1&1 &0 \\
\hline
-1&0&0
\end{array} \right)
\end{align*}
while $R_\cT$ and $R_\cL$ take the form
\begin{equation*}
R_\cT=\begin{pmatrix}
R_1
\end{pmatrix} \ \ \ R_\cL=
\begin{pmatrix}
R_2 & 0 \\
0 & R_3
\end{pmatrix}.
\end{equation*}
On this basis, we obtain
$$
H=\left( \begin{array}{c|c}
        H_{\alpha \alpha} & H_{\alpha \beta} \\
        \hline
        H_{\beta \alpha} & H_{\beta \beta}
      \end{array}
      \right)
 =\left(\begin{array}{cc|c}
    \frac{1}{R_2}+\frac{1}{R_3} & -\frac{1}{R_2} & 0 \\
    -\frac{1}{R_2} & \frac{1}{R_2} & -1 \\
    \hline
    0 & 1 & R_1
  \end{array}\right).
$$
Hence, according to (\ref{SEs FCD RLCM}) the SEs describing the circuit
in the FCD are
\begin{equation}\label{SEsChua}
   \left\{
      \begin{array}{ll}
        C_1\dot \varphi_{C1} = -\varphi_{C1}(\frac{1}{R_2}+\frac{1}{R_3})+\frac{\varphi_{C2}}{R_2}\\
        -\hat Q_M(\varphi_{C1}+\Phi_{M0})+\hat Q_M(\Phi_{M0})+C_1 v_{C10}\\
C_2\dot \varphi_{C2} = \frac{\varphi_{C1}}{R_2} -\frac{\varphi_{C2}}{R_2}+q_L+C_2v_{C20}\\
        L\dot q_L =-\varphi_{C2} -R_1q_L +Li_{L0}
      \end{array}
    \right.
\end{equation}
and, from (\ref{mixed n1}), the mixed potential is given by
\begin{align*}
\begin{split}
&\cP(\varphi_{C1},\varphi_{C2},q_L;w_0)\\
=&
\frac{1}{2}\left(\frac{1}{R_2}+\frac{1}{R_3}\right)\varphi_{C1}^2+\frac{1}{2R_2}\varphi_{C2}^2
-\frac{\varphi_{C1}\varphi_{C2}}{R_2}\\
 &-\frac{1}{2}R_1q_L^2-\varphi_{C2} q_L\\
 &+\int_0^{\varphi_{C1}} \hat Q_M(\rho+\Phi_{M0})d\rho
 -\hat Q_M(\Phi_{M0})\varphi_{C1}\\
 &-C_1v_{C01} \varphi_{C1}-C_2v_{C02} \varphi_{C2} +Li_{L0}q_L
\end{split}
\end{align*}
where $w_0=(v_{C01},v_{C02},i_{L0},\Phi_{M0})$.
It can be verified that the SEs (\ref{SEsChua}) satisfy
\begin{equation}\label{SEsChuamixed}
   \left\{
      \begin{array}{ll}
        C_1\dot \varphi_{C1} = - \frac{\partial}{\partial \varphi_{C1}}
\cP(\varphi_{C1},\varphi_{C2},q_L;w_0)\\
        C_2\dot \varphi_{C2} = - \frac{\partial}{\partial \varphi_{C2}}
\cP(\varphi_{C1},\varphi_{C2},q_L;w_0)\\
        L\dot q_L = \frac{\partial}{\partial q_L}\cP(\varphi_{C1},\varphi_{C2},q_L;w_0).
      \end{array}
    \right.
\end{equation}

\begin{figure}[t]
  \centering
\includegraphics[width=1.0\linewidth]{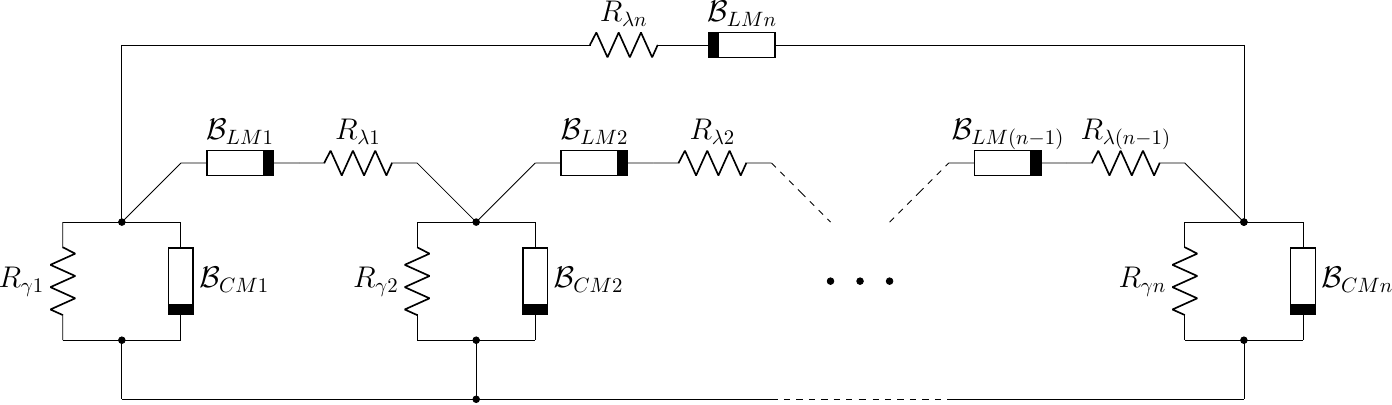}
\caption{\small 1D circular array of $n$ cells, where each cell has a $\cB_{CM}$ element
in parallel to a resistor. The interconnections between cells are given by a
$\cB_{LM}$ element in series with a resistor.}
\label{fig:arrayMC}
\end{figure}

\subsection{RLCM Circuits with a Neural Architecture}
\label{sect:neur}
\subsubsection{Array of Inductively-coupled RCM Cells}

Consider the 1D circular array with $n$ cells shown in Fig.\ \ref{fig:arrayMC}. The array has a neural-like
architecture where each cell, which is constituted by an element $\cB_{CMi}$ in parallel with a resistor $R_{\gamma i}$, $i=1,\dots,n$, is connected to the nearest-neighbor cells
via $R_{\lambda (i-1)}$ in series with $\cB_{LM(i-1)}$ and $R_{\lambda i}$ in series with $\cB_{LMi}$.
Indexes are considered modulo $n$, i.e., $\cB_{CM(n+1)}=\cB_{CM1}$
and $\cB_{CM0}=\cB_{CMn}$. It is easy to check, either by inspection, or by verifying
Assumptions\ \ref{assu:case 2 RLCM 1}, that $\varphi_{Ci}$ and $q_{Li}$, $i=1,\dots,n$,
are a complete set of variables. Application of K$q$L and K$\varphi$L yields
$$
   \left\{
      \begin{array}{ll}
        q_{CMi} = q_{L(i-1)}-q_{Li}-q_{R \gamma i}\\
        \varphi_{LMi} = \varphi_{Ci}-\varphi_{C(i+1)}-\varphi_{R \lambda i}
      \end{array}
    \right.
$$
for $i=1,\dots,n$.
By using the CRs of circuit elements (Sect.\ \ref{sect:FCD}), we obtain the SEs in the FCD
\begin{equation}\label{SEsarray}
   \left\{
      \begin{array}{ll}
        C_i \dot \varphi_{Ci} = q_{L(i-1)}-q_{Li}-\frac{\varphi_{Ci}}{R_{\gamma i}}\\
        -\hat Q_{Mi}(\varphi_{Ci}+\Phi_{M0i})+\hat Q_{Mi}(\Phi_{M0i})+C_i v_{C0i}  \\
        L_i \dot q_{Li} = \varphi_{Ci}-\varphi_{C(i+1)}-R_{\lambda i} q_{L_i}\\
        -\hat \Phi_{Mi}(q_{Li}+Q_{M0i})+\hat \Phi_{Mi}(Q_{M0i})+L_i i_{L0i}
      \end{array}
    \right.
\end{equation}
for $i=1,\dots,n$, where $n=n_C=n_L$.

Comparing with (\ref{SEs FCD RLCM}),
we obtain
$$
H  \!= \! \begin{pmatrix}
H_{\alpha \alpha} & H_{\alpha \beta}\\
\! - \! H_{\alpha \beta}^\top & H_{\beta \beta}
\end{pmatrix}
$$
where
$$
H_{\alpha \alpha}\! = \! \mathrm{diag}\left( \frac{1}{R_{\gamma 1}}, \dots,\frac{1}{R_{\gamma n}} \right); \
H_{\beta \beta} \! = \! \mathrm{diag}(R_{\lambda 1}, \dots,R_{\lambda n})
$$
and $H_{\alpha \beta}$ is the circulant matrix
\begin{align*}
\begin{split}
H_{\alpha \beta} =& \begin{pmatrix}
1 & 0 &\cdots & \cdots & 0 &-1 \\
-1&1&0 & \cdots& \cdots &0 \\
0 & -1 & 1 & 0 & \cdots & 0 \\
\vdots & \ddots & \ddots & \ddots &\ddots &\vdots\\
0&\cdots &0 &-1 & 1&0 \\
0& \cdots & \cdots & 0& -1 & 1
\end{pmatrix}\\
=& \mathrm{circ}(1, 0, \dots, 0, -1).
\end{split}
\end{align*}
Then, on the basis of (\ref{mixed n1}), we have
\begin{align*}
\begin{split}
&\cP(\varphi_{Ci},q_{Li},i=1,\dots,n;w_0)\\
=&
\frac{1}{2} \sum_{i=1}^n \frac{\varphi_{Ci}^2}{R_i}
-\frac{1}{2} \sum_{i=1}^n q_{Li}^2R_i-\sum_{i=1}^n q_{Li}(\varphi_{C(i+1)}-\varphi_{Ci})\\
 &+\sum_{i=1}^n\int_0^{\varphi_{Ci}} \hat Q_{Mi}(\rho+\Phi_{M0i})d\rho\\
 &-\sum_{i=1}^n\int_0^{\varphi_{Li}} \hat \Phi_{Mi}(\rho+Q_{M0i})d\rho\\
 &-\sum_{i=1}^n \hat Q_{Mi}(\Phi_{M0i})\varphi_{Ci}
 +\sum_{i=1}^n \hat \Phi_{Mi}(Q_{M0i})q_{Li}\\
 &-\sum_{i=1}^n C_iv_{C0i} \varphi_{Ci}+\sum_{i=1}^n L_ii_{L0i}q_{Li}
\end{split}
\end{align*}
where we let $w_0=(v_{C0i},i_{L0i},\Phi_{M0i},Q_{M0i})_{i=1,\dots,n}$.
It can be verified that the SEs can be written as follows
\begin{equation}\label{SEsarraymixed}
   \left\{
      \begin{array}{ll}
        C_i \dot \varphi_{Ci} = -\frac{\partial}{\partial \varphi_{Ci}}
\cP(\varphi_{Ci},q_{Li},i=1,\dots,n;w_0) \\
        L_i \dot q_{Li} = \frac{\partial}{\partial q_{Li}}
\cP(\varphi_{Ci},q_{Li},i=1,\dots,n;w_0)
      \end{array}
    \right.
\end{equation}
for $i=1,\dots,n$.

\begin{figure}[t]
  \centering
\includegraphics[width=0.8\linewidth]{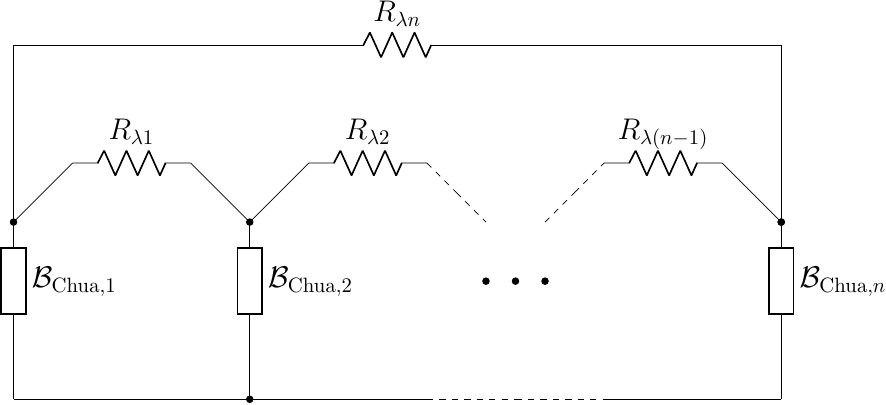}
\caption{\small 1D circular array of $n$ Chua's circuits with a memristor.
The interconnections between Chua's circuits are purely resistive.}
\label{fig:arrayChua}
\end{figure}

\subsubsection{Array of Resistively-coupled Memristor Chua's Circuits}

Figure\ \ref{fig:arrayChua} shows a circular 1D array of $n$
cells, where each cell is a Chua's circuit and the interconnections with
neighboring cells are resistors. This is a variant of the cellular nonlinear
network architecture proposed in \cite{Chual95} where, as in Section\ \ref{sect:Chua},
the nonlinear resistor in Chua's circuit is replaced by a flux-controlled
memristor. Each cell is described by elements $\cB_{CM1i}$, $\cB_{C2i}$, $\cB_{Li}$, $R_{1i}$, $R_{2i}$ and $R_{3i}$, and it is
connected to the nearest-neighbor cells via $R_{\lambda(i-1)}$ (connecting $\cB_{CM1(i-1)}$ to $\cB_{CM1i}$)  and $R_{\lambda i}$
(connecting $\cB_{CM1i}$ to $\cB_{CM1(i+1)}$). Indexes are considered modulo $n$, i.e., $\cB_{CM(n+1)}=\cB_{CM1}$
and $\cB_{CM0}=\cB_{CMn}$.
By verifying Assumptions\ \ref{assu:topo assu RLCM}, \ref{assu:case 2 RLCM 1}, it can be to checked that $\varphi_{C1i}$, $\varphi_{C2i}$ and $q_{Li}$, $i=1,\dots,n$, are a complete set of variables. Application of K$q$L and K$\varphi$L to the $i$-th cell yields
$$
   \left\{
      \begin{array}{l}
        q_{CM1i}= q_{R \lambda (i-1)}-q_{R \lambda i}+q_{R2i}-q_{R3i}\\
        q_{C2i}=q_{Li}-q_{R2i}\\
        \varphi_{Li} = -\varphi_{C2i}-\varphi_{R1i}.
      \end{array}
    \right.
$$
In the following, to simplify notations, we will assume that all cells and all resistive interconnections are identical, i.e.,
$C_{1i}=C_1$, $C_{2i}=C_2$, $L_i=L$, $M_{\Phi i}=M_\Phi$, $R_{1i}=R_1$, $R_{2i}=R_2$, $R_{3i}=R_3$ and $R_{\lambda i}=R_\lambda$ for all $i=1,\dots,n$.
By using the CRs of circuit elements (Sect.\ \ref{sect:FCD}), we obtain the SEs in the FCD
\begin{equation}\label{SEsarrayChua}
   \left\{
      \begin{array}{ll}
        C_{1}\dot \varphi_{C1i} = -\varphi_{C1i}(\frac{1}{R_{2}}+\frac{1}{R_{3}}+\frac{2}{R_\lambda})+\frac{\varphi_{C2i}}{R_{2}}\\
        +\frac{\varphi_{C1(i-1)}}{R_\lambda} +\frac{\varphi_{C1(i+1)}}{R_\lambda} -\hat Q_{M}(\varphi_{C1i}+\Phi_{M0i}) \\
        +\hat Q_{M}(\Phi_{M0i})+C_{1} v_{C10i}\\
C_{2}\dot \varphi_{C2i} = \frac{\varphi_{C1i}}{R_2} -\frac{\varphi_{C2i}}{R_2}+q_{Li}+C_{2}v_{C20i}\\
        L\dot q_{Li} =-\varphi_{C2i} -R_{1}q_{L1i}+L i_{L0i}
      \end{array}
    \right.
\end{equation}
for $i=1,\dots,n$.

Comparing with (\ref{SEs FCD RLCM}),
we obtain
$$
H= \begin{pmatrix}
H_{\alpha \alpha} & H_{\alpha \beta}\\
-H_{\alpha \beta}^\top & H_{\beta \beta}
\end{pmatrix}= \left(
\begin{array}{c|c}
\begin{array}{c|c}
H_{\alpha \alpha 11} & H_{\alpha \alpha 12} \\
\hline
H_{\alpha \alpha 12}^\top & H_{\alpha \alpha 22}
\end{array} & H_{\alpha \beta} \\
\hline
-H_{\alpha \beta}^\top & H_{\beta \beta}
\end{array}
 \right)
$$
where
\begin{align*}
H_{\alpha \alpha 11}&=\mathrm{circ}\left( \frac{1}{R_2}+\frac{1}{R_3}+\frac{2}{R_\lambda}, -\frac{1}{R_\lambda}, 0, \dots, 0, -\frac{1}{R_\lambda}\right)\\
H_{\alpha \alpha 12}&=\mathrm{diag}\left(-\frac{1}{R_2}, \dots, -\frac{1}{R_2}\right) \\
H_{\alpha \alpha 22}&=\mathrm{diag}\left(\frac{1}{R_2}, \dots, \frac{1}{R_2}\right)\\
H_{\beta \beta}&=\mathrm{diag}\left(R_{1}, \dots,R_{1}\right)
\end{align*}
and
\begin{align*}
H_{\alpha \beta}&=\begin{pmatrix}
0_n \\
-E_n
\end{pmatrix}.
\end{align*}
Then, on the basis of (\ref{mixed n1}), we have
\begin{align*}
\begin{split}
&\cP(\varphi_{C1i},\varphi_{C2i},q_{Li},i=1,\dots,n;w_0)\\
=&
\left(\frac{1}{2R_2}+\frac{1}{2R_3}+\frac{1}{R_\lambda} \right) \sum_{i=1}^n \varphi_{C1i}^2 + \frac{1}{2R_2} \sum_{i=1}^n \varphi_{C2i}^2 \\
& -\frac{1}{R_\lambda} \left(\varphi_{C11}\varphi_{C1n} +  \sum_{i=1}^{n-1} \varphi_{C1i}\varphi_{C1(i+1)} \right) \\
&-\frac{1}{R_2} \sum_{i=1}^n \varphi_{C1i} \varphi_{C2i} -\frac{R_1}{2} \sum_{i=1}^n q_{Li}^2 -\sum_{i=1}^n q_{Li}\varphi_{C2i}\\
 &+\sum_{i=1}^n\int_0^{\varphi_{C1i}} \hat Q_{M}(\rho+\Phi_{M0i})d\rho\\
 &-\sum_{i=1}^n \hat Q_{M}(\varphi_{M0i})\varphi_{C1i} - C_1\sum_{i=1}^n v_{C10i} \varphi_{C1i} \\
 &-C_2\sum_{i=1}^n v_{C20i} \varphi_{C2i}+ L \sum_{i=1}^n i_{L0i}q_{Li}
\end{split}
\end{align*}
where we let $w_0=(v_{C10i}, v_{C20i}, i_{L0i},\Phi_{M0i})_{i=1,\dots,n}$.

It can be verified that the SEs can be written as
\begin{equation}\label{SEsarrayChuamixed}
   \left\{
      \begin{array}{ll}
        C_{1}\dot \varphi_{C1i} = -\frac{\partial}{\partial \varphi_{C1i}}
        \cP(\varphi_{C1i},\varphi_{C2i},q_{Li},i=1,\dots,n;w_0) \\
		C_{2}\dot \varphi_{C2i} = -\frac{\partial}{\partial \varphi_{C2i}}
		\cP(\varphi_{C1i},\varphi_{C2i},q_{Li},i=1,\dots,n;w_0) \\
        L \dot q_{Li} = \frac{\partial}{\partial q_{Li}}
		\cP(\varphi_{Ci},q_{Li},i=1,\dots,n;w_0)
      	\end{array}
   		 \right.
\end{equation}
for $i=1,\dots,n$.



\section{Conclusion}
\label{sect:concl}
The paper has introduced a mixed potential for a large class of nonlinear
circuits, named RLCM, containing resistors, capacitors, inductors and memristors,
under a completeness assumption for the circuit variables.
This extends to nonlinear circuits with memristors the mixed potential
originally introduced by Brayton and Moser for nonlinear RLC circuits without
memristors \cite{brayton1964theoryI,brayton1964theoryII}. The extension
has been obtained by representing a memristor circuit
in the FCD and is based on a fundamental equivalence principle between an
RLCM circuit in the FCD and an RLC circuit in the VCD. Basic properties
of the mixed potential have been established. These include a compact and
useful expression of the SEs in the FCD via the mixed potential, links between
the invariants of motion in the VCD and the mixed potential as well as the
relationships between the existence of the mixed potential and the
reciprocity of multiports used in the analysis.
The results are illustrated and discussed via the application to a number
of basic memristor circuits and large-size memristor circuits
with a neural architecture. In a companion paper \cite{DiMarco2026BraytonAppl},
currently in preparation,
the mixed potential is used to systematically establish Lyapunov-like
results on convergence for RLCM circuits. Another
topic deserving future research concerns
is the possibility to further extend the
mixed potential to memristor circuits that do not satisfy the assumption of
completeness enforced in this manuscript.

\section*{Availability of data and material} The datasets generated and
analyzed during the current study are available from the corresponding author on reasonable request.

\section*{Declarations}
{\bf Conflict of interest} All authors certify that they have no affiliations with or involvement in any organization or entity with any
financial interest or non-financial interest in the subject matter or
materials discussed in this manuscript.

\bibliographystyle{unsrt}


\end{document}